\documentclass[journal]{IEEEtran}
                                                      \usepackage{blindtext}
\usepackage{graphicx}
\usepackage[utf8]{inputenc}
\usepackage{graphics} 
\usepackage{epsfig} 
\usepackage{mathptmx} 
\usepackage{amsmath} 
\usepackage{amssymb}  
\usepackage{cite}
\usepackage{physics}
\usepackage{multicol}
\usepackage{mathtools, cuted}
\usepackage[cmintegrals]{newtxmath}
\usepackage{epstopdf}
\usepackage{graphics} 
\usepackage{epsfig} 

\usepackage{mathptmx} 
\usepackage{amsmath} 
\usepackage{amssymb}  
\usepackage{cite}
\usepackage{multicol}
\usepackage{mathtools, cuted}
\usepackage[cmintegrals]{newtxmath}
\usepackage{mathtools}

\usepackage{mathptmx}
\usepackage{helvet}
\usepackage{courier}
\usepackage{subfigure}
\usepackage{type1cm}
\usepackage{titlesec}
\usepackage{epsfig} 
\usepackage{mathptmx} 
\usepackage{amsmath} 
\usepackage{amssymb}  
\usepackage{cite}
\usepackage{multicol}
\usepackage{mathtools, cuted}
\usepackage[cmintegrals]{newtxmath}
\usepackage{makeidx}         
\usepackage{algorithm}
\usepackage{algpseudocode}
\DeclareMathOperator*{\argmin}{\arg\!\min}

\usepackage{multicol}        
\usepackage[bottom]{footmisc}
\usepackage{caption}
\usepackage{nomencl}
\usepackage[usenames, dvipsnames]{color}
\usepackage{url}
\usepackage{siunitx}
\usepackage[utf8]{inputenc}

\usepackage{multicol}        
\usepackage[bottom]{footmisc}
\usepackage{caption}
\usepackage{titlesec}
\usepackage{nomencl}
\usepackage[usenames, dvipsnames]{color}
\usepackage{multirow}
\usepackage{adjustbox}
\usepackage{amsthm}
 
\theoremstyle{definition}
\newtheorem{definition}{Definition}

\newtheorem{theorem}{Theorem}

\theoremstyle{remark}
\newtheorem{remark}{Remark}
\newtheorem{proposition}{Proposition}

\newtheorem{assumption}{Assumption}

\title{\LARGE \bf
Finite-State Decentralized Policy-Based Control With Guaranteed Ground Coverage
}

\author{Hossein Rastgoftar
\thanks{Hossein Rastgoftar is with the Department of Aerospace and Mechanical Engineering, 
        Tucson, AZ 85719, USA
        {\tt\small hrastgoftar@arizona.edu}}%
}

\begin{document}

\maketitle
\thispagestyle{empty}
\pagestyle{empty}

\begin{abstract}
We propose a finite-state, decentralized decision and control framework for multi-agent ground coverage. The approach decomposes the problem into two coupled components: (i) the structural design of a deep neural network (DNN) induced by the agents’ reference configuration, and (ii) policy-based decentralized coverage control. Agents are classified as anchors and followers, yielding a generic and scalable communication architecture in which each follower interacts with exactly three in-neighbors from the preceding layer, forming an enclosing triangular communication structure. The DNN training weights implicitly encode the spatial configuration of the agent team, thereby providing a geometric representation of the environmental target set. Within this architecture, we formulate a computationally efficient decentralized Markov decision process (MDP) whose components are time-invariant except for a time-varying cost function defined by the deviation from the centroid of the target set contained within each agent’s communication triangle. By introducing the concept of Anyway Output Controllability (AOC), we assume each agent is AOC and establish decentralized convergence to a desired configuration that optimally represents the environmental target.

\end{abstract}

\section{Introduction}
Multi-agent ground coverage is a fundamental problem in distributed control with
applications in environmental monitoring, surveillance, and distributed
sensing. A classical and widely adopted approach is \emph{Voronoi-based coverage
control}, in which agents iteratively move toward the centroids of their Voronoi
cells to optimize a spatial coverage objective in a decentralized manner. This
paradigm admits strong geometric interpretability and convergence guarantees and
has been extensively studied and extended, including density-weighted coverage,
constrained environments, and event-triggered implementations
\cite{Cortes2010,Zhong2011,Kia2018}.

Despite these advantages, Voronoi-based methods typically rely on continuous-time
dynamics, frequent neighbor recomputation, and explicit geometric constructions,
which limit scalability under communication constraints and complicate
integration with discrete decision-making and learning mechanisms. These
limitations have motivated the development of \emph{policy-based and
learning-augmented decentralized frameworks}, including finite-state and Markov
decision process formulations, to address uncertainty and scalability in
multi-agent coordination \cite{Rahmani2019,Zhang2021}. While such approaches
provide increased flexibility, they often lack explicit mechanisms for encoding
formation geometry into decentralized policies or for imposing interpretable and
structured information flow with provable convergence properties.

This paper addresses these challenges by introducing a structured,
policy-based framework that tightly couples inter-agent communication,
decision-making, and physical evolution, enabling scalable decentralized
coverage while preserving geometric meaning and analytical tractability.

\subsection{Related Work}
Diffusion-based convergence and stability results for multi-agent coverage are
reported in \cite{elamvazhuthi2018nonlinear}, while decentralized coverage using
local density feedback and mean-field approximations is studied in
\cite{biswal2021decentralized}. Leader--follower coverage strategies, including
explicit separation between coordination and coverage objectives, are
investigated in \cite{atincc2020swarm}. Adaptive decentralized coverage methods
are explored in \cite{song2011decentralized,dirafzoon2011decentralized}, and
multiscale continuous-time convergence analyses are presented in
\cite{krishnan2022multiscale}. Applications to human-centered sensing and zone
coverage planning appear in \cite{9147790,9336858}. A substantial body of work adopts Voronoi-based coverage control
\cite{bai2021adaptive,nguyen2016discretized,abbasi2017new,luo2019voronoi},
typically establishing convergence via Lyapunov-based arguments under kinematic
or single-integrator agent abstractions. Extensions addressing obstacles,
failures, and leader--follower structures are considered in \cite{bai2021adaptive},
while experimental comparisons in complex urban environments are reported in
\cite{patel2020multi}. Coverage control for heterogeneous agent teams has also received increasing
attention. Authors of \cite{santos2018heterogeneous} propose a heterogeneous coverage control framework that encodes qualitatively different sensing capabilities through agent-specific density functions in a locational cost, deriving a distributed gradient-descent controller with additional boundary terms that ensures convergence to critical points of the heterogeneous coverage objective and demonstrates improved performance over heterogeneous Lloyd-type methods in experiments. A Voronoi-based coverage control method for heterogeneous disk-shaped robots, leveraging power diagrams and constrained centroidal motion to ensure collision-free convergence to locally optimal sensing configurations, is proposed in \cite{arslan2013heterogeneous}. Sadeghi and Smith address coverage control for multiple event types with heterogeneous robots by formulating an event-specific Voronoi partitioning framework and deriving distributed algorithms with provable convergence to locally optimal sensing configurations in both continuous and discrete environments \cite{sadeghi2019event}. A coverage control framework for robots with heterogeneous maximum speeds is presented in \cite{kim2022speed}, formulating a temporal cost based on multiplicatively weighted Voronoi diagrams and deriving a gradient-based controller that yields time-optimal coverage configurations. More recent work considers multi-resource and persistent surveillance objectives \cite{coffey2023multi,hua2025persistent}.

Learning-based approaches have formulated multi-agent coverage as a decision
process using reinforcement learning and Markov decision models
\cite{lauer2000decentralized,olfati2002mdp,din2022deep,dai2020graph,xiao2020distributed}.
While these methods offer scalability and adaptability, they typically rely on
unstructured communication, large or continuous state spaces, and
gradient-based optimization, limiting interpretability and convergence analysis.

To clarify the distinction from existing learning-based coverage approaches, we
emphasize that this paper does not treat multi-agent coverage as a generic
reinforcement learning or function approximation problem. Instead, inter-agent
communication is explicitly architected through a hierarchical anchor--follower
structure induced by a reference configuration, yielding unidirectional,
feedforward information flow. This structure allows the multi-agent system
itself to be interpreted as a dynamical neural network whose neurons correspond
to physical agents and whose activations are governed by agent dynamics rather
than algebraic mappings. Learning is performed via forward-only, local updates
without gradient backpropagation or centralized critics, and each agent solves a
finite, time-invariant local Markov decision process defined geometrically within
its communication triangle. These features fundamentally distinguish the
proposed framework from existing RL- and MDP-based coverage methods.

\subsection{Contributions}
This paper proposes a \emph{policy-based, decentralized framework} for coverage
of unknown ground targets that scales to arbitrarily large teams and is
independent of individual agent dynamics. The key idea is to reinterpret
multi-agent coverage as a \emph{structured dynamical system} in which
communication, decision-making, and physical evolution are intrinsically
coupled. By organizing inter-agent communication according to a reference
configuration, the proposed approach induces a hierarchical, feedforward
coordination architecture that admits a dynamical deep neural network
interpretation while remaining fully decentralized.

The main contributions are summarized as follows:
\begin{itemize}
    \item \textbf{Structured Communication and Dynamical DNN Representation:}
    A hierarchical anchor--follower communication architecture is introduced
    that induces unidirectional, feedforward information flow, enabling the
    multi-agent system to be interpreted as a dynamical neural network whose
    neurons correspond to physical agents.

    \item \textbf{Forward-Only Learning with Local Observability:}
    Communication weights are learned using exclusively forward, local updates
    without gradient backpropagation, centralized critics, or global information.

    \item \textbf{Decentralized Policy Learning via Local MDPs:}
    Each follower agent independently learns a transition policy by solving a
    finite, time-invariant local Markov decision process defined geometrically
    within its communication triangle.

    \item \textbf{Dynamics-Agnostic Coverage via Anyway Output Controllability:}
    The notion of \emph{Anyway Output Controllability} decouples policy learning
    from specific agent dynamics, enabling uniform application to heterogeneous
    teams with nonlinear, underactuated, or high-order dynamics.
\end{itemize}

\subsection{outline}
This paper is organized as follows: The problem statement is presented in Section \ref{sec:problem_statement}. An algorithmic approach for structuring the DNN based on the agent team’s reference configuration  is developed in Section~\ref{Structuring of the DNN and Specifying the DNN Weights}. Training the DNN weighsts is defined as and MDP and presented in Section \ref{Training the DNN Weights}. Stability and convergence of the proposed policy-based decentralized coverage solution are proven in Section~\ref{Multi-Agent Coverage Dynamics and Control}. Simulation results are presented in Section~\ref{Simulation Results}, followed by the conclusion in Section~\ref{Conclusion}.

\begin{figure}
    \centering
    \includegraphics[width=0.48\textwidth]{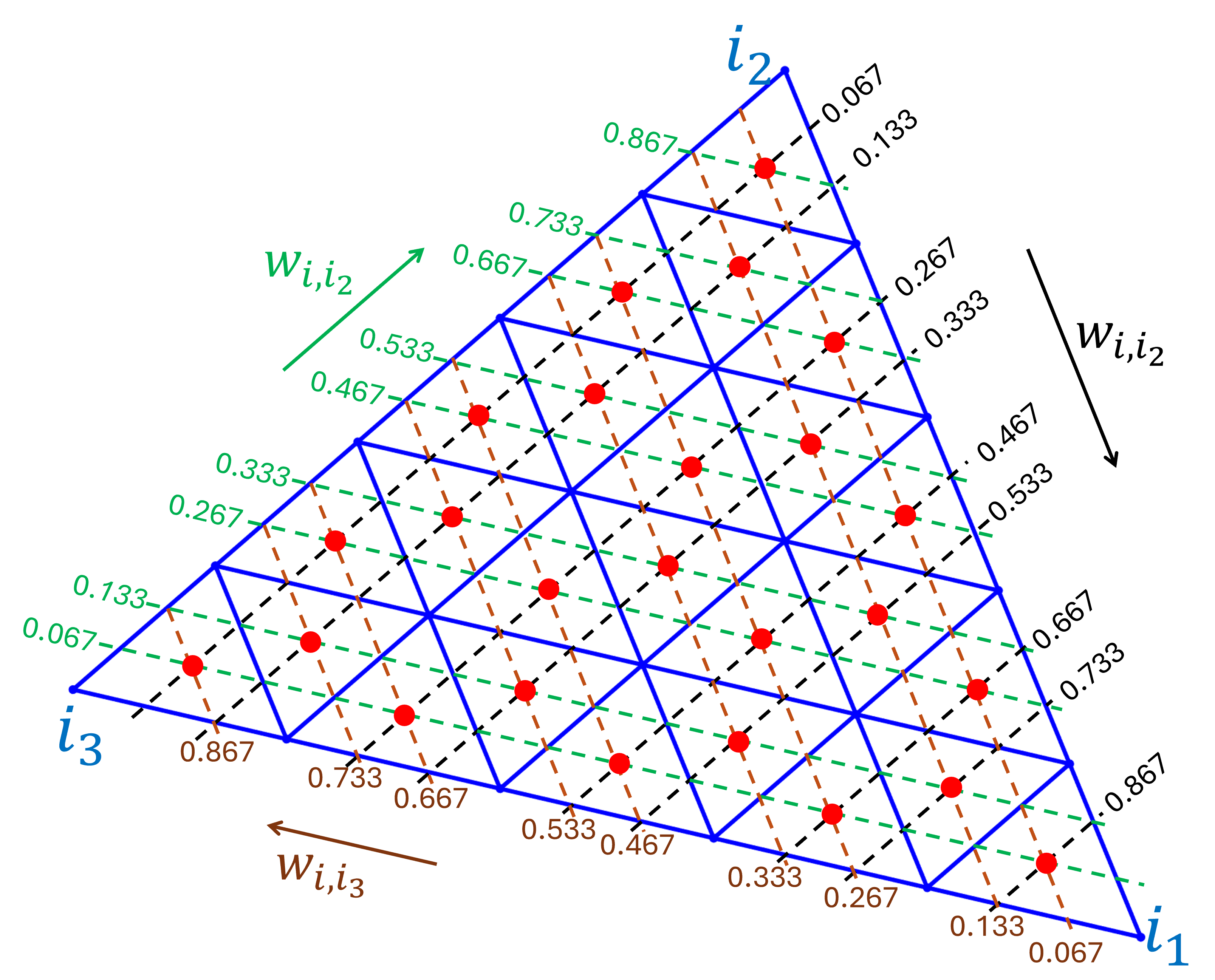}
    \vspace{-0.25cm}
    \caption{Geometric representation of the first- and second-tier communication weights,
$w_{i,i_1}\in\mathcal{W}_{i,1}$ and $w_{i,i_2}\in\mathcal{W}_{i,2}$, for $M_i=5$,
where agent $i\in\mathcal{V}\setminus\mathcal{V}_0$ has three in-neighbors
$\mathcal{N}(i)=\{i_1,i_2,i_3\}$.
}
    \label{States-Weights}
\end{figure}
\section{Problem Statement}\label{sec:problem_statement}

We consider a team of $N$ agents indexed by
\[
\mathcal{V}=\{1,\ldots,N\},
\]
tasked with providing aerial coverage of a finite set of ground targets
$\mathcal{D}$.  Agents are classified as \emph{boundary} or \emph{interior}
according to a reference configuration, and their interactions are structured
via a Delaunay neighbor network (DNN) to enable scalable coverage.
Low-level control dynamics are abstracted, and each agent is assumed to satisfy
the following output reachability property.

\begin{definition}[Anyway Output Controllability (AOC)]\label{def: AOC}
Let agent $i\in\mathcal{V}$ be described by
\begin{equation}
\begin{cases}
\mathbf{x}_i[t+1] = \mathbf{f}_i\!\left(\mathbf{x}_i[t],\mathbf{u}_i[t]\right),\\
\mathbf{r}_i[t] = \mathbf{h}_i\!\left(\mathbf{x}_i[t]\right),
\end{cases}
\end{equation}
where $\mathbf{x}_i$, $\mathbf{u}_i$, and $\mathbf{r}_i$ denote the state, input,
and output, respectively.
Agent $i$ is said to be \emph{Anyway Output Controllable} if, for any initial
state $\mathbf{x}_i[t]\in\mathcal{X}_i$, there exist an admissible input sequence
$\mathbf{u}_i(\cdot)$ and a finite time
$T_i(\mathbf{x}_i[t])<\infty$ such that
\[
\mathbf{r}_i\!\left[t+T_i(\mathbf{x}_i[t])\right]\in\mathcal{P}_i,
\]
where $\mathcal{X}_i\subset\mathbb{R}^2$ and
$\mathcal{P}_i\subset\mathbb{R}^2$ are compact sets.
\end{definition}

\begin{assumption}\label{assum1}
The time discretization is chosen uniformly across agents and sufficiently
large such that the output reachability time satisfies
\[
T_i(\mathbf{x}_i[t]) = 1,
\]
for all $i\in\mathcal{V}$ and all initial states
$\mathbf{x}_i[t]\in\mathcal{X}_i$. Consequently,
\[
\mathbf{r}_i[t+1]\in\mathcal{P}_i,
\]
holds for any admissible initial condition.
\end{assumption}
The objective of this paper is to design a decentralized framework that enables structured agent interactions and adaptive coverage of distributed targets. Specifically, we address the following problems.

\subsection*{Problem 1 (DNN Structuring).}
Given a reference configuration, a deterministic algorithm uniquely induces a DNN communication architecture from the agents’ initial spatial distribution. The agent set $\mathcal{V}$ is partitioned into $M+1$ disjoint subsets
\[
\{\mathcal{V}_l\}_{l=0}^M, \qquad \bigcup_{l=0}^M \mathcal{V}_l = \mathcal{V},
\]
where $\mathcal{V}_0$ consists of anchor nodes, and each agent $i \in \mathcal{V}_l$, $l \ge 1$, has exactly three in-neighbors.

\subsection*{Problem 2 (Decentralized Coverage via Discrete DNN Weights).}
Design a decentralized control and learning mechanism that enables the agent team to achieve high-level coverage of the distributed target set $\mathcal{D}$.  For each agent $i\in\mathcal{V}\setminus\mathcal{V}_0$, the DNN training weights are restricted to finite discrete sets
\begin{equation}\label{DiscreteWeights}
    \mathcal{W}_i=\left\{{3a-b\over 3M_i}: a=1,\cdots,M_i,~b=1,2\right\}
\end{equation}
where $M_i\in\mathbb{N}$ determines the discretization resolution. These set $\mathcal{W}_i$ consists of uniformly distributed values in $(0,1)$, ensuring strictly positive and bounded training weights. 
For each agent $i\in\mathcal{V}\setminus\mathcal{V}_0$, let $\mathcal{N}(i)$ denote its set of in-neighbors. 
The communication weight between agent $i$ and neighbor $j\in\mathcal{N}(i)$ is denoted by $w_{i,j}$ and satisfies
\begin{subequations}\label{weightcond}
\begin{equation}
w_{i,j}\in\mathcal{W}_i, \qquad \forall\, i\in\mathcal{V}\setminus\mathcal{V}_0,\ \forall\, j\in\mathcal{N}(i),
\end{equation}
\begin{equation}
\sum_{j\in\mathcal{N}(i)} w_{i,j} = 1, \qquad \forall\, i\in\mathcal{V}\setminus\mathcal{V}_0.
\end{equation}
\end{subequations}

For clarity, Fig.~\ref{States-Weights} illustrates the geometric representation of the communication weights for $M_i=5$, where agent $i\in\mathcal{V}\setminus\mathcal{V}_0$ interacts with three in-neighbors,
$\mathcal{N}(i)=\{i_1,i_2,i_3\}$.
The corresponding discrete weight set
\[
\resizebox{0.95\hsize}{!}{%
$
\mathcal{W}_i=\{0.067,0.133,0.267,0.333,0.467,0.533,0.667,0.733,0.867\}
$
}
\]
is obtained from~\eqref{DiscreteWeights}. Problem 2 formulates  the coverage problem as a decentralized MDP with time-invariant state space, action space, state transition model, and discount factor, and a time-varying cost function capturing coverage performance. The detailed MDP formulation and DNN weight training procedure are presented in Section~\ref{Training the DNN Weights}.

\vspace{-0.3cm}
\section{Structuring of the  Coverage DNN}\label{Structuring of the DNN and Specifying the DNN Weights}
The DNN communication architecture is induced by partitioning the agent set $\mathcal{V}$, based on a reference configuration, into $M+1$ disjoint groups indexed by $\mathcal{M} := \{0,1,\ldots,M\}$. This induces the decomposition
\[
\mathcal{V} = \bigcup_{l\in\mathcal{M}} \mathcal{V}_l, 
\qquad 
\mathcal{V}_l \cap \mathcal{V}_h = \emptyset,\; l \neq h,
\]
with $\mathcal{V}_l \subset \mathcal{V}$ and cardinality $|\mathcal{V}_l| = N_l$ for all $l \in \mathcal{M}$.
 Define the cumulative index
\[
P_l =
\begin{cases}
\sum_{h=0}^{l} N_h, & l\in\mathcal{M}\setminus\{0\},\\
0, & l=0,
\end{cases}
\]
and index the agents by $\{b_1,\ldots,b_N\}$. Then
\begin{equation}\label{vl}
    \mathcal{V}_l=\{b_{P_{l-1}+1},\ldots,b_{P_l}\}.
\end{equation}


To define inter-agent communication, introduce the nested sets
\begin{equation}\label{wl}
\mathcal{L}_l =
\begin{cases}
\mathcal{V}_l, & l\in\{0,M\},\\
\mathcal{V}_l \cup \mathcal{L}_{l-1}, & \text{otherwise},
\end{cases}
\qquad \forall l\in\mathcal{M}.
\end{equation}
Let $\mathcal{I}(i,l)\subseteq \mathcal{L}_{l-1}$ denote the set of neurons in layer $l-1$ connected to neuron $i \in \mathcal{L}_l$. The DNN architecture is induced from the agents’ initial formation via the algorithmic procedure in Algorithm~\ref{euclid33}, which constructs a directed graph $\mathcal{G}(\mathcal{V},\mathcal{E})$ that admits a DNN representation. In particular, $\mathcal{G}(\mathcal{V},\mathcal{E})$ determines: (i) the number of DNN layers $(M+1)$, (ii) a partition of $\mathcal{V}$ into subsets $\mathcal{V}_0,\ldots,\mathcal{V}_M$, and (iii) the inter-layer neuron connectivity.

Given $\mathcal{E} \subset \mathcal{V}\times\mathcal{V}$, the in-neighbor set of agent $i\in\mathcal{V}$ is defined as
\begin{equation}
    \mathcal{N}(i) := \{\, j \in \mathcal{V} \mid (j,i) \in \mathcal{E} \,\}.
\end{equation}
Then, for each layer $l \in \mathcal{M}\setminus\{0\}$, the interconnection set $\mathcal{I}(i,l)$ for neuron $i \in \mathcal{L}_l$ is given by
\begin{equation}
    \mathcal{I}(i,l)=
    \begin{cases}
        \mathcal{N}(i), & i \in \mathcal{V}_l = \mathcal{L}_l \setminus \mathcal{L}_{l-1},\\[2pt]
        \{i\}, & i \in \mathcal{L}_l \setminus \mathcal{V}_l,
    \end{cases}
    \qquad l \in \mathcal{M}\setminus\{0\}.
\end{equation}
\begin{remark}
Algorithm~\ref{euclid33} applies to decentralized multi-agent systems in $\mathbb{R}^n$; the ground coverage setting considered here corresponds to $n=2$.
\end{remark}

\begin{figure*}[h]
\centering
\subfigure[$l=0\in \mathcal{M}$]{\includegraphics[width=0.23\linewidth]{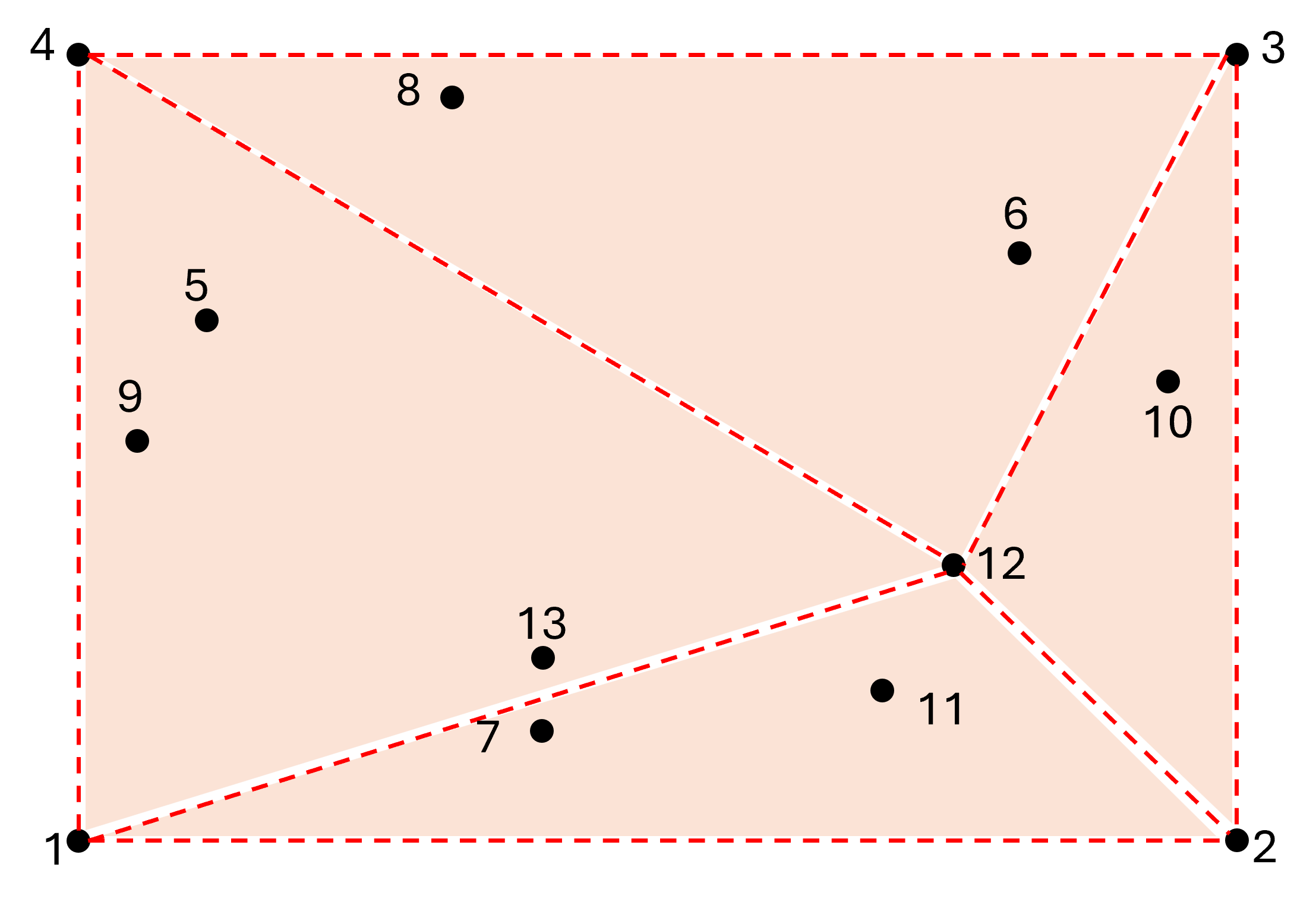}}
\subfigure[$l=1\in \mathcal{M}$]{\includegraphics[width=0.23\linewidth]{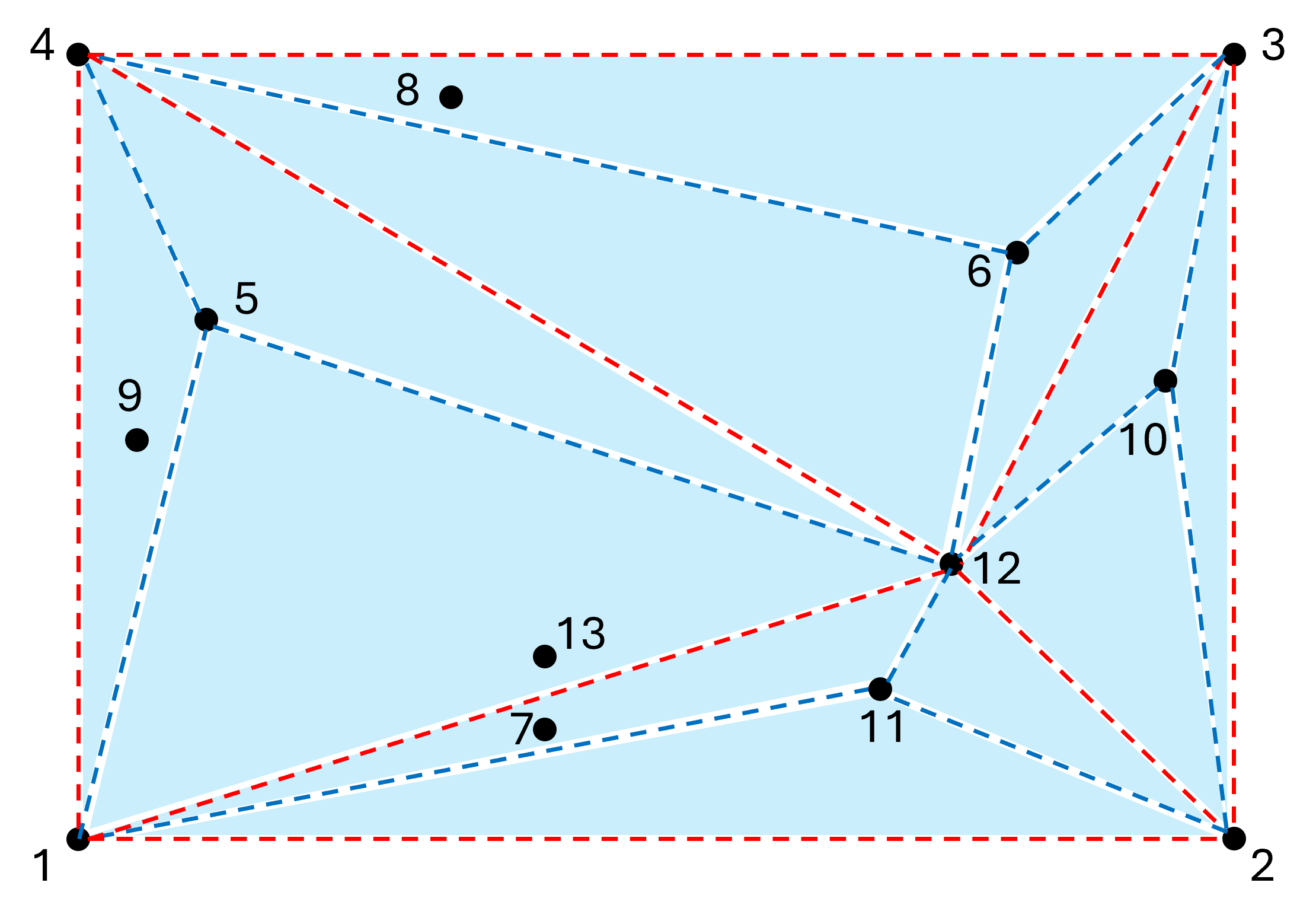}}
\subfigure[$l=2\in \mathcal{M}$]{\includegraphics[width=0.23\linewidth]{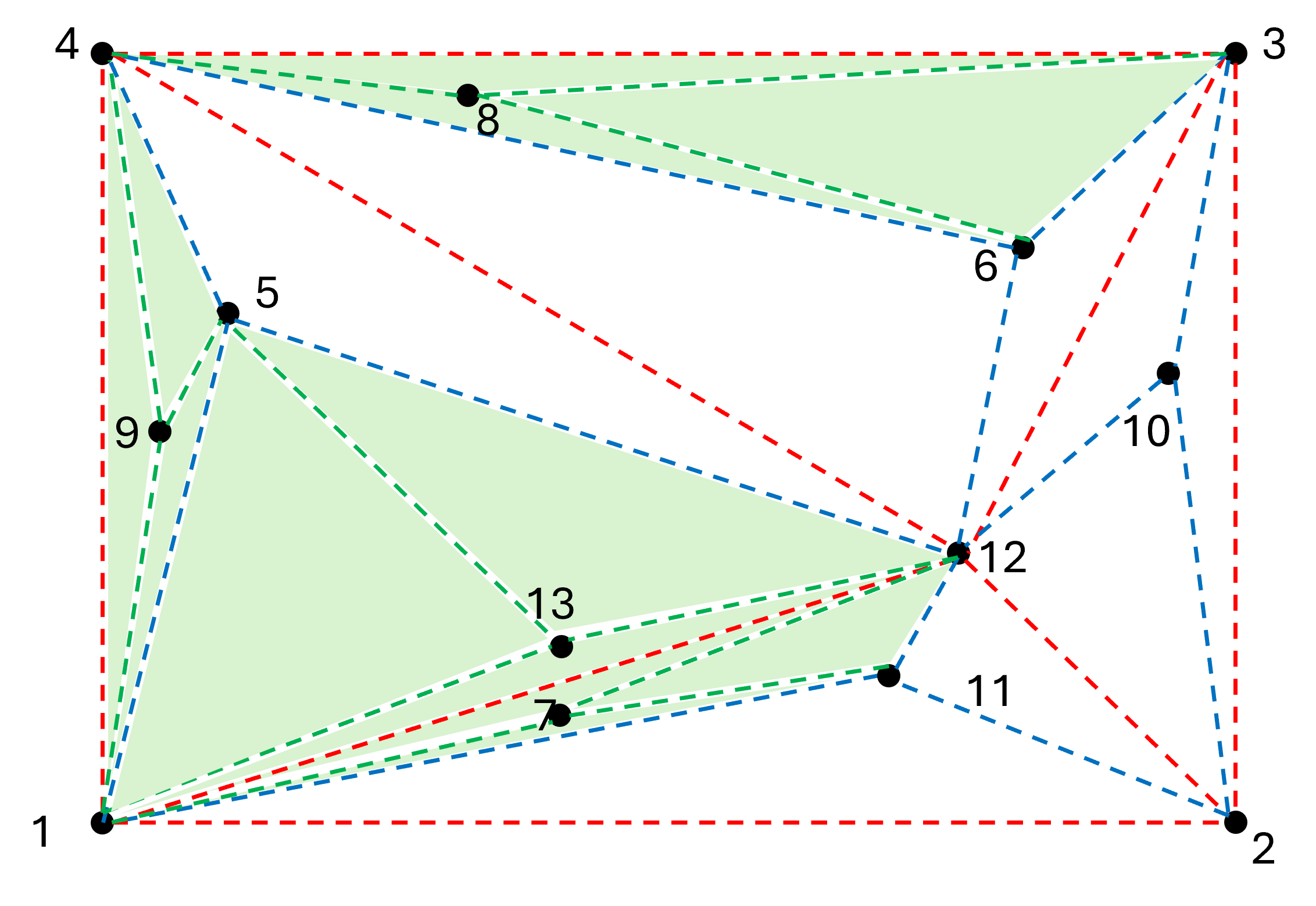}}
\subfigure[$\mathcal{M}=\left\{0,1,2\right\}$]{\includegraphics[width=0.28\linewidth]{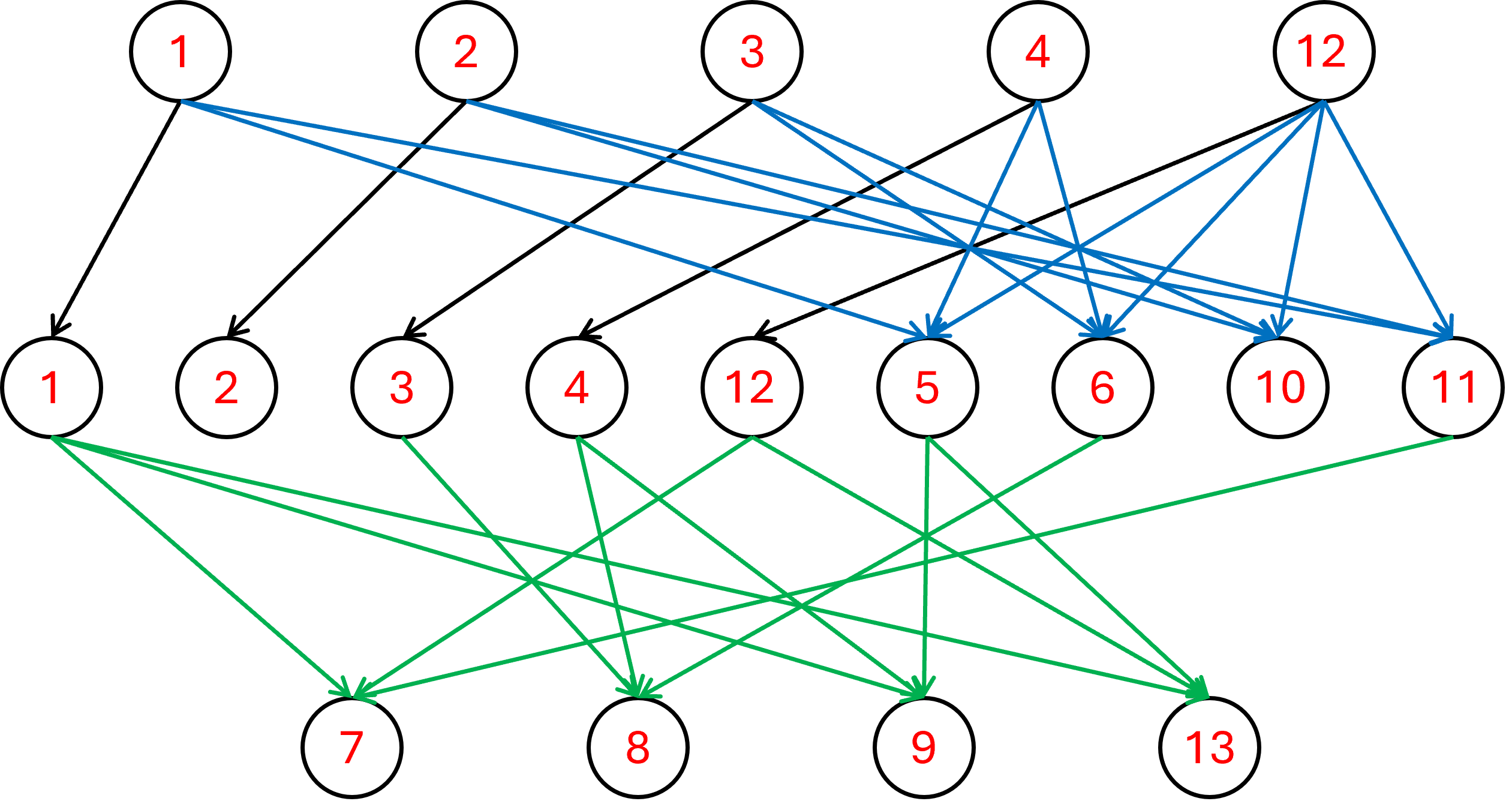}}
\vspace{-0.4cm}
\caption{Cell decompositions of the convex hull defined by the boundary agents for specifying DNN layer interconnections.
}
\label{CellDecomposition}
\end{figure*}

For each agent $i \in \mathcal{V}$, the following position-related quantities are used throughout the paper:
\begin{itemize}
    \item $\mathbf{a}_i$: reference position of agent $i$ in the initial (reference) configuration.
    \item $\mathbf{r}_i[t]$: actual position of agent $i$ at discrete time $t$, given by the output of its control system.
    \item $\mathbf{c}_i[t]$: reference input to the control system of agent $i$ at time $t$; for $i \in \mathcal{V}_0$, $\mathbf{c}_i[t]$ is constant, while for $i \in \mathcal{V}\setminus\mathcal{V}_0$, it is defined as a weighted average of the actual positions of its in-neighbor agents.
    \item $\mathbf{p}_i$: desired position of agent $i$; $\mathbf{p}_i$ is known for $i \in \mathcal{V}_0$ and unknown for $i \in \mathcal{V}\setminus\mathcal{V}_0$.
\end{itemize}
Moreover, for all $i \in \mathcal{V}\setminus\mathcal{V}_0$, the reference input satisfies $\mathbf{c}_i[t] = \mathbf{p}_i$.

\subsection{Step 1: Agent Classification} The agent set $\mathcal{V}$ is decomposed as $\mathcal{V} = \mathcal{V}_B \cup \mathcal{V}_I$, where $\mathcal{V}_B$ and $\mathcal{V}_I$ are disjoint. The set $\mathcal{V}_B = \{b_1,\ldots,b_{N_B}\}$ consists of the boundary agents located at the vertices of the convex hull enclosing the interior agents in $\mathcal{V}_I$. The polytope defined by $\mathcal{V}_B$ is referred to as the \emph{leading polytope}. Given the boundary set $\mathcal{V}_B$, the core agent is identified using one of two criteria:
(i) the interior agent minimizing the aggregate distance to the boundary agents,
\begin{equation}\label{coreagent}
b_{N_B+1}
=
\arg\min_{i \in \mathcal{V}\setminus \mathcal{V}_B}
\sum_{j \in \mathcal{V}_B}
\|\mathbf{a}_i - \mathbf{a}_j\|,
\end{equation}
or (ii) the agent located near the center of the target domain $\mathcal{D}$.

Given the set $\mathcal{V}_B$ and the designated core agent $b_{N_B+1}$, we define the set $\mathcal{V}_0$ as $\mathcal{V}_0=\mathcal{V}_B\cup \left\{b_{N_B+1}\right\}$. According to Eq.~\eqref{wl}, we have $\mathcal{V}_0 = \mathcal{L}_0$. The leading polytope can be partitioned into $m_0$ distinct simplex cells. Consequently, the set $\mathcal{L}_0$ can be expressed as a union of these simplices:
\begin{equation}
    \mathcal{L}_0=\bigcup_{h=1}^{m_0}\mathcal{R}_{0,h},
\end{equation}
where $\mathcal{R}_{0,h}$ determines vertices of the $h$-th simplex cell of the leading polytope. For better clarification, an agent team with  $N=13$ agents forms a $2$-dimensional formation  shown in Fig. \ref{CellDecomposition} (a), where $\mathcal{V}_B=\left\{1,\cdots,4\right\}$ ($N_B=4$) defines the boundary agents. Agent $b_5=12\in \mathcal{V}$ is assigned by \eqref{coreagent} as the core leader, therefore, $\mathcal{V}_0=\mathcal{L}_0=\left\{1,\cdots,4,12\right\}$ defines agents of the first layer. The convex hull defined by  $\mathcal{L}_0$ can be decomposed into $m_0=4$ triangular cells with vertices defined by  $\mathcal{R}_{0,1}=\left\{1,2,12\right\}$, $\mathcal{R}_{0,2}=\left\{2,3,12\right\}$, $\mathcal{R}_{0,3}=\left\{3,4,12\right\}$, and $\mathcal{R}_{0,4}=\left\{4,1,12\right\}$.

\subsection{Step 2: Expansion and Structuring} Set $\mathcal{V}$ can be expressed as $\mathcal{V}=\mathcal{L}_{l-1}\cup \bar{\mathcal{L}}_{l-1}$, for every $l\in \mathcal{M}\setminus\left\{0\right\}$, where  $\bar{\mathcal{L}}_{l-1}=\mathcal{V}\setminus \mathcal{L}_{l-1}$ defines the agents not belonging to $\mathcal{L}_{l-1}$.  Note that $\mathcal{V}_l\subset \bar{\mathcal{L}}_{l-1}$, if $\bar{\mathcal{L}}_{l-1}\neq \emptyset$. Also, $\mathcal{L}_{l-1}$ consists of $m_{l-1}$ distinct simplices that cover the domain contained by $\mathcal{L}_{l-1}$. Therefore, $\mathcal{L}_{l-1}$ can be expressed as:
\begin{equation}
    \mathcal{L}_{l-1}=\bigcup_{h=1}^{m_{l-1}}\mathcal{R}_{l-1,h},\qquad l\in \mathcal{M}\setminus \left\{0\right\}
\end{equation}
where $\mathcal{R}_{l-1,1}$ through $\mathcal{R}_{l-1,{m_{l-1}}}$ are vertices of distinct simplex cells that cover the domain contained by $\mathcal{L}_{l-1}$. Given a set $\mathcal{R}_{l-1,h}$ for each $h = 1, \dots, m_{l-1}$, we denote by $\mathrm{CONV}(\mathcal{R}_{l-1,h})$ the convex hull formed by the elements of $\mathcal{R}_{l-1,h}$. We also define $\mathcal{H}_{l-1,h} \subset \bar{\mathcal{L}}_{l-1}$ as the set of all nodes that lie within this convex hull, i.e., all nodes contained in $\mathrm{CONV}(\mathcal{R}_{l-1,h})$.

If $\mathcal{H}_{l-1,h}\neq \emptyset$, then:
\begin{itemize}
    \item $\mathcal{R}_{l-1,h}$ has a mentee that is determined by:
\begin{equation}\label{ah}
        \mu_{l-1,h}=\mathop{\arg\min}_{j\in \mathcal{H}_{l-1,h}}\left(\sum_{r\in \mathcal{R}_{l-1,h}}\|\mathbf{a}_r-\mathbf{a}_j\|\right),~h=1,\cdots,m_{l-1}.
    \end{equation}
    \item In-neighbors of $\mu_{l-1,h}\in\mathcal{L}_l$ is defined by $\mathcal{N}\left(\mu_{l-1,h}\right)=\mathcal{R}_{l-1,h}$.
\end{itemize}

Note that the mentee of $\mathcal{R}_{l-1,h}$, denoted by $\mu_{l-1,h}$, does not exist if $\mathcal{H}_{l-1,h}=\emptyset $). Then, for every $l\in \mathcal{M}\setminus \left\{0\right\}$, $\mathcal{V}_l$ aggengates the mentess of all non-empty simplices of $\mathcal{L}_{l-1}$ and defined as follows:
\begin{equation}
\begin{split}
    \mathcal{V}_l=&\bigg\{i\in \mathcal{H}_{l-1,h}:\mathcal{H}_{l-1,h}\neq \emptyset,~ i=\argmin_{j\in \mathcal{H}_{l-1,h}\neq \emptyset }\left(\sum_{r\in \mathcal{R}_{l-1,h}}\|\mathbf{a}_r-\mathbf{a}_j\|\right),\\
    &h=1,\cdots,m_l-1\bigg\}.
\end{split}    
\end{equation}
Therefore, for every $l \in \mathcal{M} \setminus \{0\}$, the number of agents in $\mathcal{V}_l$ satisfies $N_l = |\mathcal{V}_l| \leq m_{l-1}$. This inequality holds because not all simplices in $\mathcal{L}_{l-1}$ necessarily have mentee agents assigned to them. By knowing $\mathcal{V}_l$ and $\mathcal{L}_{l-1}$, $\mathcal{L}_l$ is defined by Eq. \eqref{wl}.

\subsection{Step 3: Cell Decomposition Update} If $\mathcal{H}_{l-1,h}\neq \emptyset $, then, $\mu_{l-1,h}$ exists  and $CONV(\mathcal{R}_{l-1,h})$ can be decomposed into $n+1$ new simplex  cells all sharing $\mu_{l-1,h}$. Therefore, the leading polytope is deterministically decomposed into $m_l$ distinct simplices by knowing $\mathcal{V}_l$, where $m_l\leq (n+1)m_{l-1}$.

\begin{algorithm}
  \caption{DNN Structure based on reference formation.}\label{euclid33}
  \begin{algorithmic}[1]
        \State \textit{Get:} Agents' initial positions $\mathbf{a}_1$ through $\mathbf{a}_{N}$
        \State \textit{Obtain:} Edge set $\mathcal{E}$, $M=\left|\mathcal{M}\right|$, and $\mathcal{V}_0$ through $\mathcal{V}_M$.
        \State Assign boundary agents $\mathcal{V}_B=\left\{b_1,\cdots,b_{N_B}\right\}$.
        \State Assign core agent $b_{N_B+1}$ using Eq. \eqref{coreagent}.
        \State Define $\mathcal{V}_0=\mathcal{V}_B\cup\left\{b_{N_B+1}\right\} $.
        \State Define $\mathcal{L}_0=\mathcal{V}_0$ and $\bar{\mathcal{L}}_{0}=\mathcal{V}\setminus \mathcal{L}_0$.
        \State Decompose the leading polytope into $m_0$ simplex cells with vertices defined by $\mathcal{R}_{0,1}$, $\cdots$, and $\mathcal{R}_{0,m_0}$.
        \State $l=1$.      
        \While{$\bar{\mathcal{L}}_{l-1}\neq \emptyset$}
            \State $m_l=0$, $N_l=0$, and $\mathcal{V}_l=\emptyset$;
            \For{\texttt{< $ h=1,\cdots,m_{l-1}$>}}
                \If{$CONV(\mathcal{R}_{l-1,h})$ contains at least one agent}
                    \State Assign mentee of $\mathcal{R}_{l-1,h}$, denoted by $\mu_{l-1,h}$;  
                    \State Define neighbors of $\mu_{l-1,h}$: $\mathcal{N}\left({\mu_{l-1,h}}\right)=\mathcal{R}_{l-1,h}$;
                    \State $N_l\leftarrow N_l+1$;
                    \State $m_l\leftarrow m_l+n+1$;
                    \State $\mathcal{V}_l=\mathcal{V}_l\cup \left\{\mu_{l-1,h}\right\}$;
                    \State Specify $\mathcal{R}_{l,(n+1)\left(N_l-1\right)+1}$, $\cdots$, and $\mathcal{R}_{l,(n+1)N_l}$
                \EndIf
            \EndFor 
            \State Obtain $\mathcal{L}_l$, and $\bar{\mathcal{L}}_l$.
            \State $l\leftarrow l+1$.
        \EndWhile    
        \State $M=l-1$.
  \end{algorithmic}
\end{algorithm}

For better clarification, Fig. \ref{CellDecomposition} shows how Algorithm \ref{euclid33} is implemented to specify the inter-agent communication based on the agent team reference configuration. As shown in Fig. \ref{CellDecomposition} (a), $CONV\left(\mathcal{R}_{0,h}\right)$ is a traingular cell that contains at least one agent for $h=1,\cdots,4$. Therefore, $m_1=12$ and $CONV\left(\mathcal{R}_{0,h}\right)$ is decomposed into three tringular cells shown in Fig. \ref{CellDecomposition} (a), each shown in blue. For layer $l=1$, $\mathcal{V}_1=\left\{11,10,6,5\right\}$ where $11$, $10$, $6$, and $5$ are mentors of $\mathcal{R}_{0,1}$, $\mathcal{R}_{0,2}$, $\mathcal{R}_{0,3}$, and $\mathcal{R}_{0,4}$, respectively. Also, $m_1=12$; 
 $\mathcal{R}_{1,1}=\left\{1,2,11\right\}$, $\mathcal{R}_{1,2}=\left\{2,12,11\right\}$, $\mathcal{R}_{1,3}=\left\{12,1,11\right\}$,  $\mathcal{R}_{1,4}=\left\{2,3,10\right\}$, $\mathcal{R}_{1,5}=\left\{3,12,10\right\}$,  $\mathcal{R}_{1,6}=\left\{12,2,10\right\}$,  $\mathcal{R}_{1,7}=\left\{3,4,6\right\}$,   $\mathcal{R}_{1,8}=\left\{4,12,6\right\}$,   $\mathcal{R}_{1,9}=\left\{12,3,6\right\}$,   $\mathcal{R}_{1,10}=\left\{4,1,5\right\}$,   $\mathcal{R}_{1,11}=\left\{1,12,5\right\}$,   and $\mathcal{R}_{1,12}=\left\{4,5,12\right\}$ define the verities of $12$ tringular cells shown in Fig. \ref{CellDecomposition} (b). As shown in Fig. \ref{CellDecomposition} (b), $CONV\left(\mathcal{R}_{1,h}\right)$ contains a single agent if $h=3,7,10,11$ and $CONV\left(\mathcal{R}_{1,h}\right)$ does not contain an agent otherwise. Therefore, $\mathcal{V}_2=\left\{7,8,9,13\right\}$, and $7$, $8$, $9$, and $13$ are mentees of $\mathcal{R}_{1,3}=\mathcal{N}(7)$, $\mathcal{R}_{1,7}=\mathcal{N}(8)$, $\mathcal{R}_{1,10}=\mathcal{N}(9)$, and $\mathcal{R}_{1,11}=\mathcal{N}({13})$ for $l=2$ (see Fig. \ref{CellDecomposition} (c)). Because $\bar{\mathcal{L}}=\emptyset$, the while loop of Algorithm \ref{euclid33} stops at $l=M=2$, and as a result, the DNN shown in Fig. \ref{CellDecomposition} (d) specifies the inter-agent communications.

\begin{figure*}[h]
\centering
\subfigure[]{\includegraphics[width=0.24\linewidth]{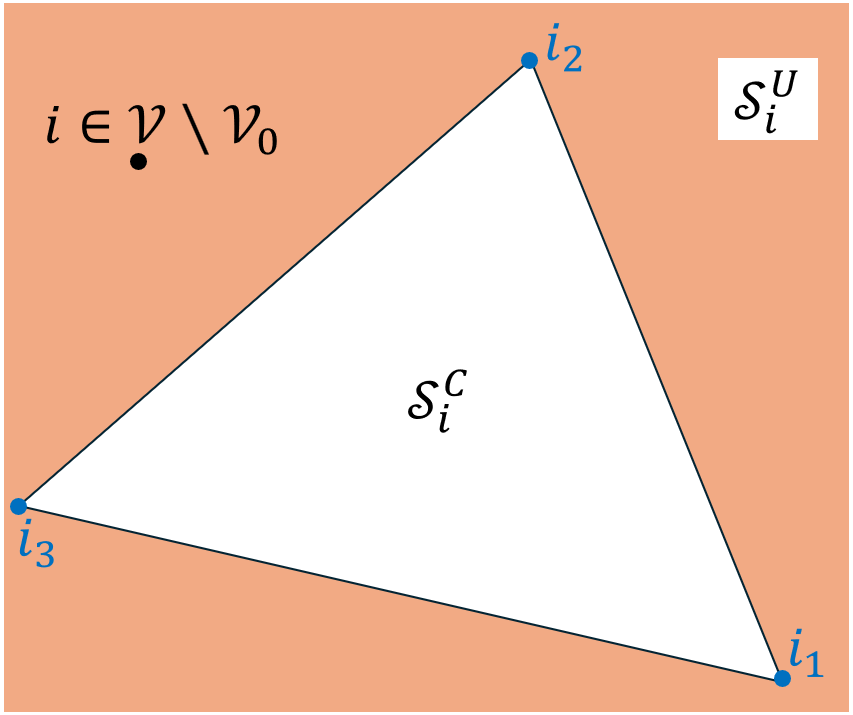}}
\subfigure[]{\includegraphics[width=0.24\linewidth]{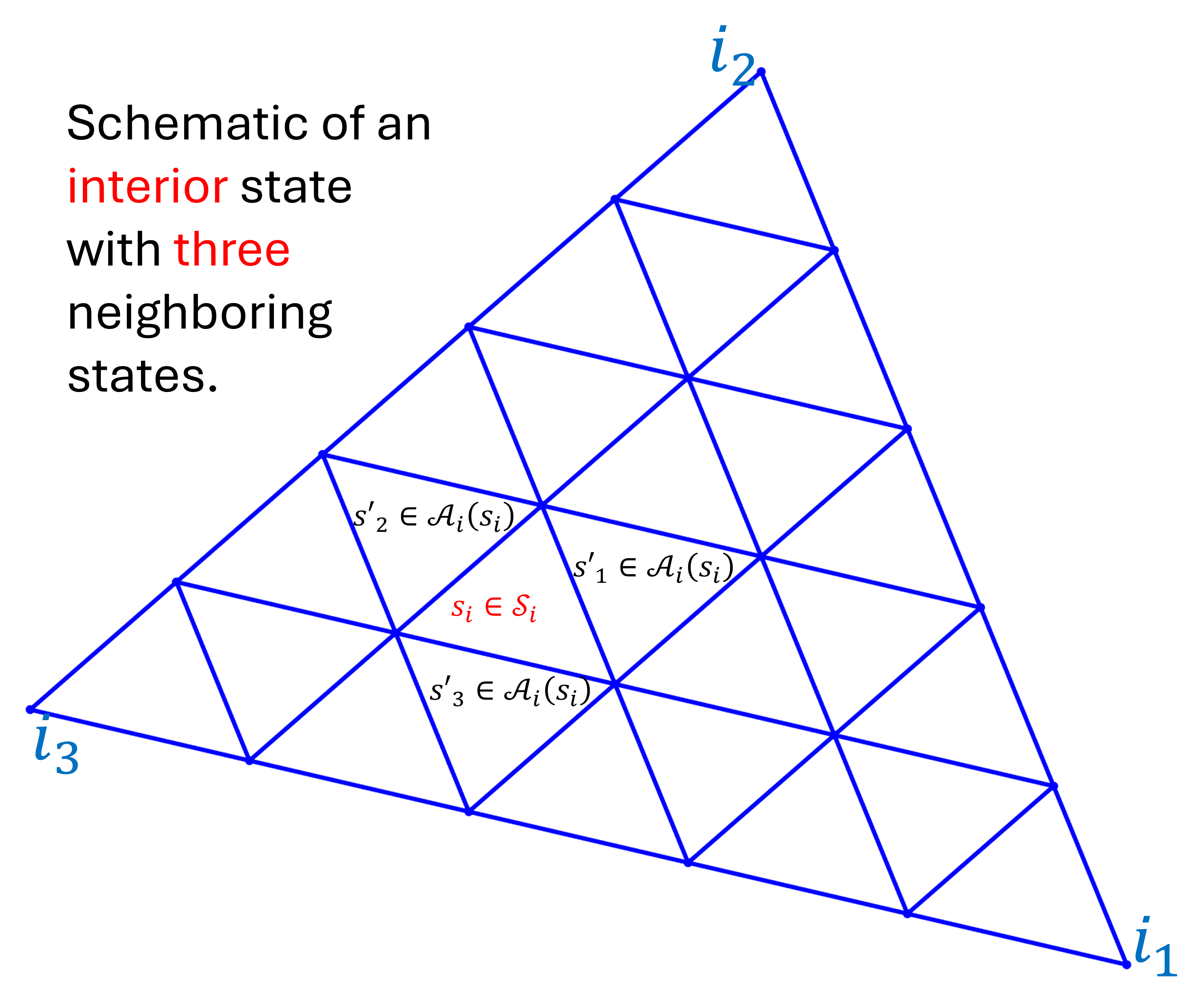}}
\subfigure[]{\includegraphics[width=0.24\linewidth]{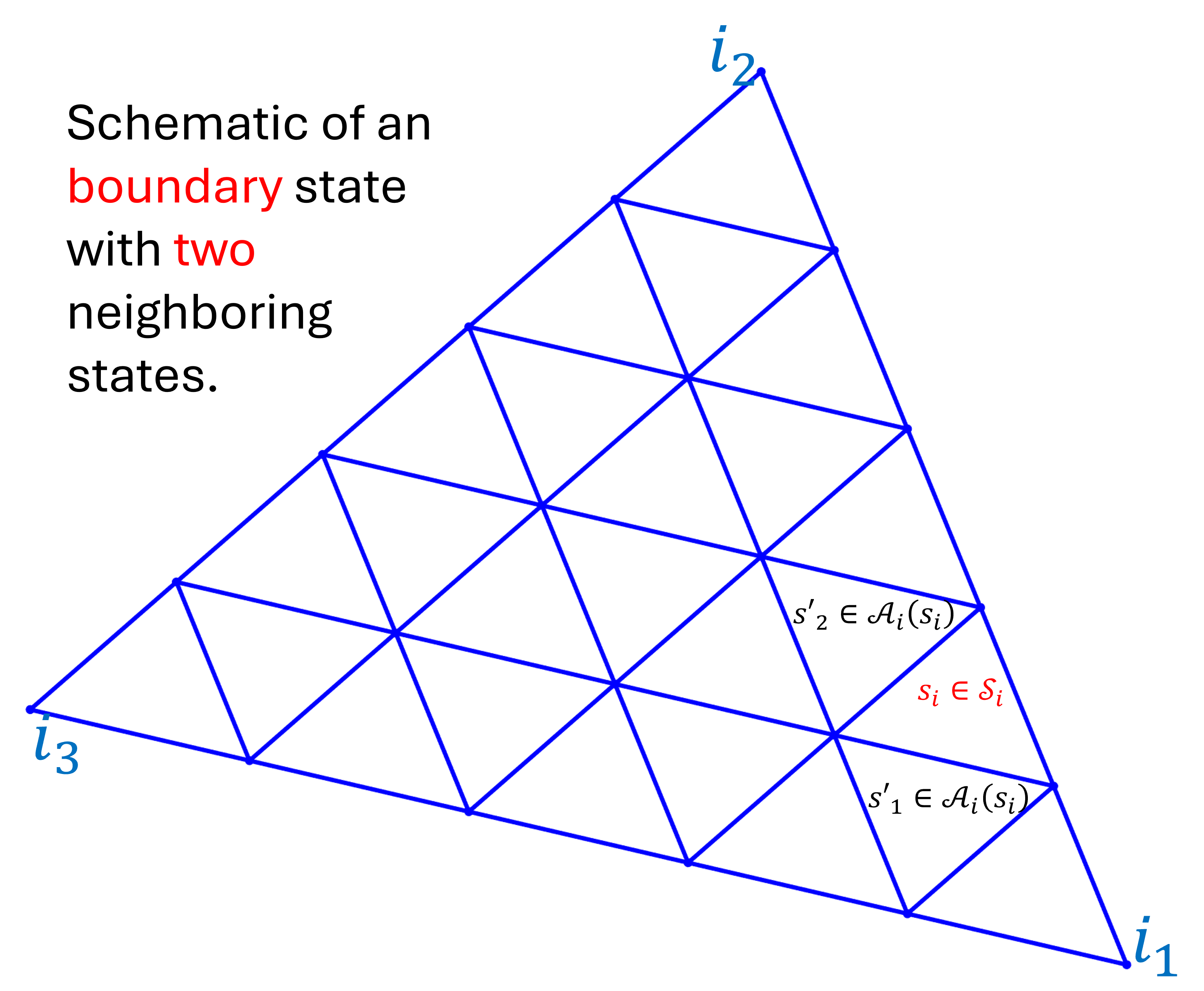}}
\subfigure[]{\includegraphics[width=0.24\linewidth]{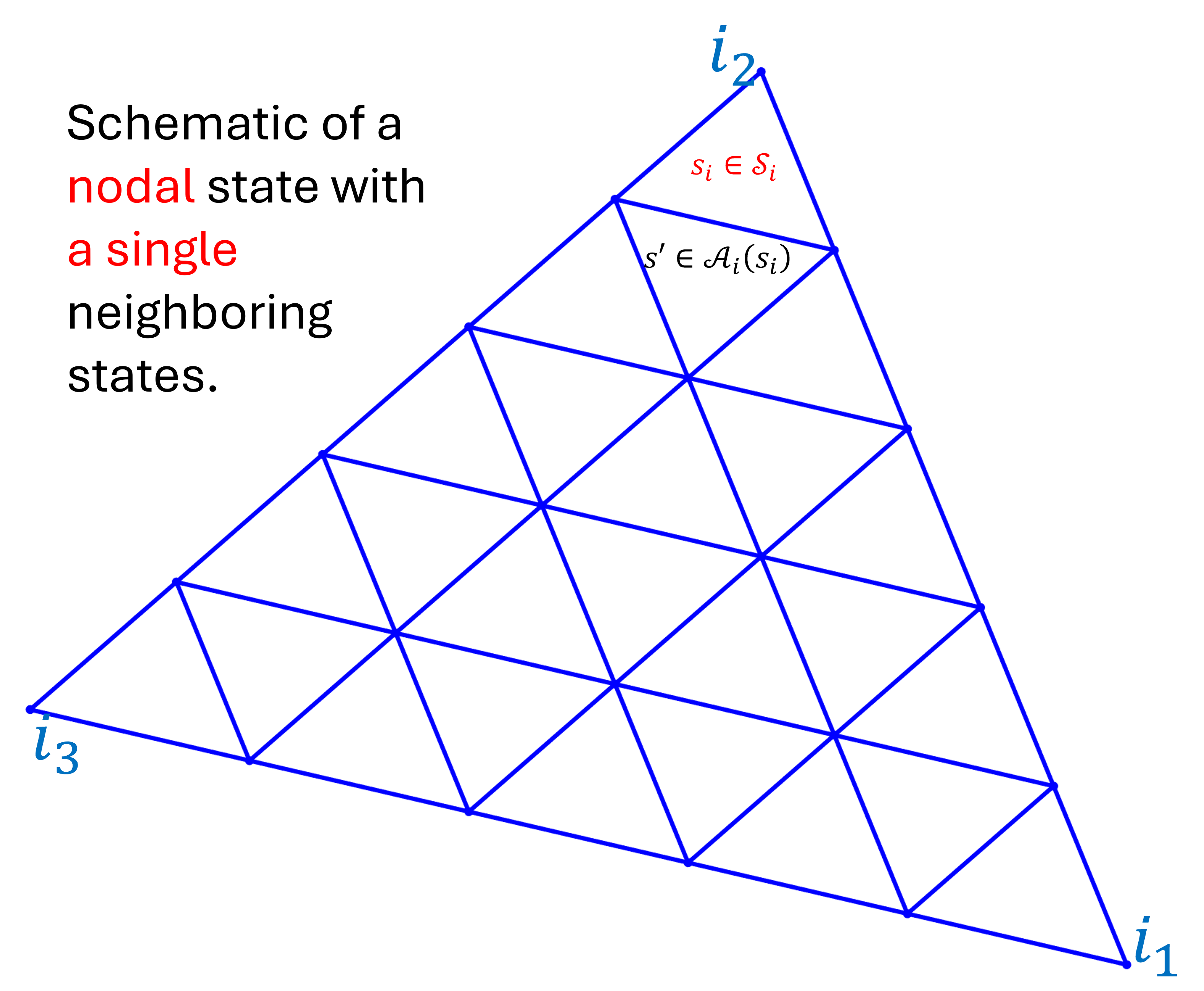}}
\vspace{-0.4cm}
\caption{(a) Schematic of $\mathcal{S}_i^U$, shown by light red,  defining a single state outside the communication triangle of agent $i\in \mathcal{V}\setminus \mathcal{V}_0$. Schematic illustration of (b) an interior state, (c) a boundary state, and (d) a nodal state $s\in\mathcal{S}_i^C$, having three, two, and one neighboring states.
}
\label{States-Neighbors}
\end{figure*}

\section{Training the DNN Weights}\label{Training the DNN Weights}

The DNN is trained in a \emph{fully decentralized and agent-centric} manner, wherein each agent
$i \in \mathcal{V}\setminus\mathcal{V}0$ independently optimizes its local communication strategy.
Specifically, agent~$i \in \mathcal{V}\setminus\mathcal{V}0$ assigns adaptive communication weights
$w_{i,j}\in\mathcal{W}_i$ to its in-neighbors $j\in\mathcal{N}(i)$ by solving a local MDP whose components and operation are described in Sections~\ref{MDP Components} and~\ref{Operation}, respectively.
This formulation enables scalable and communication-aware learning without centralized coordination, while guaranteeing coverage convergence via the theoretical results established in Section~\ref{Multi-Agent Coverage Dynamics and Control}.
\subsection{MDP Components}\label{MDP Components}
Agent~$i\in\mathcal{V}\setminus\mathcal{V}_0$ is associated with an MDP defined as
\[
\mathcal{M}_i
\big(
\mathcal{S}_i,\mathcal{D}_i,\mathcal{A}_i,g_i,
P_i,\mathcal{C}_i,
\gamma_i
\big),
\]
where the components of $\mathcal{M}_i$ are detailed below.

\noindent\textbf{State Set:}
The state space of agent~$i\in\mathcal{V}\setminus\mathcal{V}_0$ is defined by $\mathcal{S}_i$ and partitioned as
\begin{equation}
\mathcal{S}_i=\mathcal{S}_i^{C}\dot{\cup}\mathcal{S}_i^{U},
\end{equation}
where $\mathcal{S}_i^{C}$ and $\mathcal{S}_i^{U}$ denote the \emph{contained} and
\emph{uncontained} subspaces, respectively. For each agent $i\in\mathcal{V}\setminus\mathcal{V}_0$, the contained subspace
$\mathcal{S}_i^{C}$ is obtained by discretizing the communication triangle
$\mathcal{T}_i[t]$ formed by the instantaneous positions of its in-neighbors
$\mathcal{N}(i)$. Specifically,
\[
\mathcal{S}_i^{C}=\{s_1,\ldots,s_{M_i^2}\},
\]
where $M_i=\lvert\mathcal{W}_i\rvert$ (see~\eqref{DiscreteWeights}), yielding
$\lvert\mathcal{S}_i^{C}\rvert=M_i^2$. Each state $s\in\mathcal{S}_i^{C}$ corresponds
to a triangular cell with centroid
\begin{equation}
\mathbf{c}_i(s)=\sum_{j\in\mathcal{N}(i)} w_{i,j}(s)\mathbf{r}_j,\qquad s\in \mathcal{S}_i^C,
\end{equation}
where the barycentric weights $w_{i,j}\in\mathcal{W}_i$ satisfy~\eqref{weightcond}.
The uncontained subspace $\mathcal{S}_i^{U}$ is represented by a single aggregate
state capturing all positions outside $\mathcal{T}_i[t]$.
Fig.~\ref{States-Neighbors} illustrates the resulting discretization.

\noindent\textbf{Local Target Set:}
Let $\mathbf{d}_j:\mathcal{D}\rightarrow\mathbb{R}^2$ denote the position of
environmental target $j\in\mathcal{D}$. The set of targets locally observable by
agent~$i\in\mathcal{V}\setminus\mathcal{V}_0$ is defined as
\begin{equation}
\mathcal{D}_i[t]=\{j\in\mathcal{D}:\mathbf{d}_j[t]\in\mathcal{T}_i[t]\},
\qquad
\forall i\in\mathcal{V}\setminus\mathcal{V}_0.
\end{equation}

\noindent\textbf{Action Set:}
The action space $\mathcal{A}_i$ encodes admissible transitions over
$\mathcal{S}_i$ and is defined as a mapping
$\mathcal{A}_i:\mathcal{S}_i\rightarrow \mathcal{S}_i$.
Two triangular cells are considered neighbors if they share a common edge.
Accordingly, each state in $\mathcal{S}_i^{C}$ admits at most three neighboring
cells, resulting in at most four actions (including self-transition).
If $s\in\mathcal{S}_i^{C}$, then $\mathcal{A}_i(s)\subseteq\mathcal{S}_i^{C}$,
whereas actions from $s\in\mathcal{S}_i^{U}$ transition into the contained subspace.
Fig.~\ref{States-Neighbors}(b)--(d) illustrates this construction.

\noindent\textbf{Goal State:}
The goal state $g_i\in\mathcal{S}_i^{C}$ is selected to maximize the coverage
quality of the local target set $\mathcal{D}_i$. Define
\begin{equation}
{\mathbf{h}}_i[t]=
\begin{cases}
\dfrac{1}{|\mathcal{D}_i[t]|}\displaystyle\sum_{j\in\mathcal{D}_i[t]} \mathbf{d}_j[t],
& \mathcal{D}_i[t]\neq\emptyset,\\[8pt]
\dfrac{1}{3}\displaystyle\sum_{j\in\mathcal{N}(i)} \mathbf{r}_j[t],
& \mathcal{D}_i[t]=\emptyset,
\end{cases}
\end{equation}
which represents the centroid of the locally sensed targets when available, and
otherwise the centroid of the communication triangle.
The goal state $g_i$ is then defined as the unique triangular cell containing
${\mathbf{h}}_i[t]$.

\noindent\textbf{Transition Dynamics:}
The transition kernel is modeled as a linear combination of fixed base transition
measures. Let
$\Phi_i:\mathcal{S}_i\times\mathcal{A}_i\rightarrow\mathbb{R}^d$
denote a feature map, and let $\{\mu_{i,j}\}_{j=1}^d$ be graph-constrained base
transition distributions. The resulting transition kernel is
\begin{equation}\label{transitionke}
P_i(\cdot\mid s,a)
=
\sum_{j=1}^d \phi_{i,j}(s,a)\,\mu_{i,j}(\cdot),
\qquad\forall i\in\mathcal{V}\setminus\mathcal{V}_0.
\end{equation}

\noindent\textbf{Cost Function:}
Let $\bar{\mathbf{r}}(s)$ denote the centroid of the triangular cell associated with
state $s\in\mathcal{S}_i$. The MDP employs a state-dependent cost function defined as
\begin{equation}
\mathcal{C}_i(s)=
\begin{cases}
\alpha\big\|\bar{\mathbf{r}}(s)-\bar{\mathbf{r}}(g_i)\big\|, & s\neq g_i,\\[4pt]
\alpha\big\|\bar{\mathbf{r}}(s)-\bar{\mathbf{r}}(g_i)\big\|-\beta, & s=g_i,
\end{cases}
\end{equation}
where $\alpha>0$ and $\beta>0$ are design parameters. The term proportional to $\alpha$
penalizes deviation from the goal state $g_i$, while the terminal reward $\beta$
incentivizes reaching the goal.

\noindent\textbf{Discount Factor:}
The parameter $\gamma_i\in(0,1)$ denotes the discount factor.

\subsection{Operation}\label{Operation}
An agent $i\in \mathcal{V}\setminus \mathcal{V}_0$ may lie either inside or
outside its communication triangle $\mathcal{T}_i[t]$ at time $t$.
If the agent position satisfies $\mathbf{r}_i[t]\notin \mathcal{T}_i[t]$,
we set $M_i=1$. In this case, the corresponding MDP state satisfies
$s\in\mathcal{S}_i^{U}$, the contained subspace reduces to a singleton
$\lvert\mathcal{S}_i^{C}\rvert=1$, and the action mapping
$\mathcal{A}_i(s)=\mathcal{S}_i^{C}$ assigns a single admissible successor
state to the uncontained state. Consequently, the optimal action is trivial
and no Bellman recursion is required.

In contrast, when $s\in\mathcal{S}_i^{C}$, the admissible action set satisfies
$\mathcal{A}_i(s)\subseteq\mathcal{S}_i^{C}$. In this case, the optimal value
function is computed via the Bellman optimality equation
\begin{equation}\label{OptimalValue}
\resizebox{0.99\hsize}{!}{%
$
V_i^*(s)
=
\min\limits_{a\in \mathcal{A}_i(s)}
\left\{
\mathcal{C}_i(s,a)
+
\gamma
\sum\limits_{s'\in \mathcal{S}_i}
P_i(s'\mid s,a)\,V_i^*(s')
\right\},
\quad
s\in\mathcal{S}_i^{C},
$
}
\end{equation}
with the corresponding optimal policy given by
\begin{equation}\label{OptimalPolicy}
\resizebox{0.99\hsize}{!}{%
$
\pi_i^*(s)
=
\arg\min\limits_{a\in \mathcal{A}_i(s)}
\left\{
\mathcal{C}_i(s,a)
+
\gamma
\sum\limits_{s'\in \mathcal{S}_i}
P_i(s'\mid s,a)\,V_i^*(s')
\right\},
\quad
s\in\mathcal{S}_i^{C}.
$
}
\end{equation}

\section{Stability and  Convergence}\label{Multi-Agent Coverage Dynamics and Control}
In this section, we provide the proofs for the stablility and converegence of the proposed decentralized coverage method.
\begin{definition}
Let $\mathcal{V}=\{b_1,\ldots,b_N\}$. Define
\begin{equation}
\mathbf{y}
=\mathrm{vec}\!\left(
\begin{bmatrix}
\mathbf{r}_{b_1} & \cdots & \mathbf{r}_{b_N}
\end{bmatrix}^{\!T}
\right)
\in \mathbb{R}^{2N}.
\end{equation}
\end{definition}

\begin{definition}
Let $\mathcal{V}_l$ be defined as in \eqref{vl}. Define
\begin{equation}
\mathbf{y}_l
=\mathrm{vec}\!\left(
\begin{bmatrix}
\mathbf{r}_{b_{P_{l-1}+1}} & \cdots & \mathbf{r}_{b_{P_l}}
\end{bmatrix}^{\!T}
\right)
\in \mathbb{R}^{2N_l}.
\end{equation}
\end{definition}

\begin{assumption}\label{currentaoutside}
For any $s\in\mathcal{S}_i^{U}$, the target set $\mathcal{P}_i$ is a triangle \emph{strictly contained} in the communication triangle $\mathcal{T}_i(t)$ and \emph{edge-aligned} with $\mathcal{T}_i(t)$.
\end{assumption}

\begin{assumption}\label{currentinside}
For any $s\in\mathcal{S}_i^{C}$, the target set $\mathcal{P}_i$ is a triangle \emph{strictly contained} in the target triangle induced by the optimal next state $\pi_i^*(s)$ and \emph{edge-aligned} with it.
\end{assumption}

\begin{theorem}\label{thm:markov}
Assume each agent $b_i\in\mathcal{V}\setminus\mathcal{V}_0$ satisfies the AOC property and, in Assumption~\ref{assum1},
the target set $\mathcal{P}_{b_i}$ is replaced by the time-varying communication triangle $\mathcal{T}_{b_i}[t]$.
Then the coverage evolution satisfies
\begin{equation}\label{eq:markov_update}
\mathbf{y}[t+1]=\mathbf{\Gamma}[t]\mathbf{y}[t],
\end{equation}
where $\mathbf{\Gamma}[t]$ is row-stochastic for all $t$. Consequently, \eqref{eq:markov_update} defines a time-inhomogeneous Markov process.
\end{theorem}

\begin{proof}
For each anchored agent $b_i\in\mathcal{V}_0$, $\mathbf{r}_{b_i}[t]=\mathbf{p}_{b_i}$ for all $t$.
For any $b_i\in\mathcal{V}\setminus\mathcal{V}_0$, the AOC property and the transition rules in Section~\ref{Operation} imply
\[
\begin{split}
    \mathbf{r}_{b_i}[t+1]
=&\sum_{b_j\in\mathcal{N}(b_i)}\alpha_{b_i,b_j}[t]\mathbf{r}_{b_j}[t],\\
\alpha_{b_i,b_j}[t]\ge&0,\\
\sum_{b_j\in\mathcal{N}(b_i)}\alpha_{b_i,b_j}[t]=&1.
\end{split}
\]
Stacking all agent positions yields \eqref{eq:markov_update} with entries
\[
\Gamma_{i,j}[t]=
\begin{cases}
1, & b_i\in\mathcal{V}_0,\ i=j,\\
\alpha_{b_i,b_j}[t], & b_i\in\mathcal{V}\setminus\mathcal{V}_0,\ b_j\in\mathcal{N}(b_i),\\
0, & \text{otherwise}.
\end{cases}
\]
Each row of $\mathbf{\Gamma}[t]$ is nonnegative and sums to one; hence $\mathbf{\Gamma}[t]$ is row-stochastic.
\end{proof}

\begin{theorem}\label{thm:anchored_stability}
Consider $\mathbf{y}[t+1]=\mathbf{\Gamma}[t]\mathbf{y}[t]$, where each $\mathbf{\Gamma}[t]\in\mathbb{R}^{N\times N}$ is row-stochastic.
Let $\mathcal{V}_0$ denote anchored agents satisfying $\mathbf{r}_{b_i}[t]\equiv\mathbf{p}_{b_i}$ for all $b_i\in\mathcal{V}_0$ and all $t$.
After reordering agents, write
\begin{equation}\label{eq:block_partition}
\mathbf{\Gamma}[t]=
\begin{bmatrix}
\mathbf{I} & \mathbf{0}\\
\mathbf{B}[t] & \mathbf{A}[t]
\end{bmatrix},
\qquad
\mathbf{y}[t]=
\begin{bmatrix}
\mathbf{p}\\
\mathbf{y}_F[t]
\end{bmatrix}.
\end{equation}
Assume there exist $T\ge1$ and $\eta\in(0,1)$ such that:
\begin{enumerate}
\item[(C1)] Assumptions~\ref{currentaoutside}--\ref{currentinside} hold.
\item[(C2)] For every $t$ and every follower index $i$,
\begin{equation}\label{eq:reachability}
\sum_{j\in\mathcal{V}_0}
\big(\mathbf{\Gamma}[t+T-1]\cdots\mathbf{\Gamma}[t]\big)_{ij}
\ge \eta .
\end{equation}
\end{enumerate}
Then the follower subsystem
\begin{equation}\label{eq:follower_update}
\mathbf{y}_F[t+1]=\mathbf{A}[t]\mathbf{y}_F[t]+\mathbf{B}[t]\mathbf{p}
\end{equation}
is globally exponentially stable: for all $t\ge s$,
\begin{equation}\label{eq:exp_contraction}
\|\mathbf{\Phi}(t,s)\|_\infty
=\|\mathbf{A}[t-1]\cdots\mathbf{A}[s]\|_\infty
\le (1-\eta)^{\left\lfloor\frac{t-s}{T}\right\rfloor}.
\end{equation}
Moreover, $\mathbf{y}_F[t]$ converges, and each follower coordinate converges to a convex combination of the anchors' coordinates.
\end{theorem}

\begin{proof}
From \eqref{eq:block_partition} and anchor invariance, the follower dynamics are \eqref{eq:follower_update}.
Row-stochasticity of $\mathbf{\Gamma}[t]$ implies $\mathbf{A}[t]\ge0$ and $\mathbf{A}[t]\mathbf{1}\le\mathbf{1}$, i.e., $\mathbf{A}[t]$ is substochastic.

Define the $T$-step product
\[
\begin{split}
    \mathbf{M}[t]=&\mathbf{\Gamma}[t+T-1]\cdots\mathbf{\Gamma}[t]
=
\begin{bmatrix}
\mathbf{I} & \mathbf{0}\\
\mathbf{B}_T[t] & \mathbf{A}_T[t]
\end{bmatrix},\\
\qquad
\mathbf{A}_T[t]=&\mathbf{A}[t+T-1]\cdots\mathbf{A}[t].
\end{split}
\]
Since $\mathbf{M}[t]$ is row-stochastic, for each follower row $i\in \mathcal{V}\setminus \mathcal{V}_0$,
\[
\sum_{j\in\mathcal{V}_0}(\mathbf{M}[t])_{ij}+\sum_{\ell}(\mathbf{A}_T[t])_{i\ell}=1.
\]
By (C2), $\sum_{j\in\mathcal{V}_0}(\mathbf{M}[t])_{ij}\ge\eta$, hence $\sum_{\ell}(\mathbf{A}_T[t])_{i\ell}\le1-\eta$.
Therefore,
\begin{equation}\label{eq:AT_contract}
\|\mathbf{A}_T[t]\|_\infty=\max_i\sum_{\ell}(\mathbf{A}_T[t])_{i\ell}\le 1-\eta.
\end{equation}

Let $\mathbf{\Phi}(t,s)=\mathbf{A}[t-1]\cdots\mathbf{A}[s]$. Grouping the product into blocks of length $T$ and using submultiplicativity of $\|\cdot\|_\infty$ with \eqref{eq:AT_contract} gives \eqref{eq:exp_contraction}, proving exponential contraction of the homogeneous system.

Unrolling \eqref{eq:follower_update} yields
\[
\mathbf{y}_F[t]=\mathbf{\Phi}(t,0)\mathbf{y}_F[0]+\sum_{\tau=0}^{t-1}\mathbf{\Phi}(t,\tau+1)\mathbf{B}[\tau]\mathbf{p}.
\]
Because rows of $[\mathbf{B}[t]\ \mathbf{A}[t]]$ sum to $1$, $\|\mathbf{B}[t]\|_\infty\le1$, and \eqref{eq:exp_contraction} implies $\|\mathbf{\Phi}(t,\tau+1)\|_\infty$ decays geometrically, the series converges; hence $\mathbf{y}_F[t]$ converges.

Finally, each follower update is a convex combination of neighbor states and fixed anchors, so each follower coordinate remains in the convex hull of the anchors (and the shrinking follower contribution), and the limit is a convex combination of the anchors' coordinates.
\end{proof}

\begin{figure}
    \centering
    \includegraphics[width=0.48\textwidth]{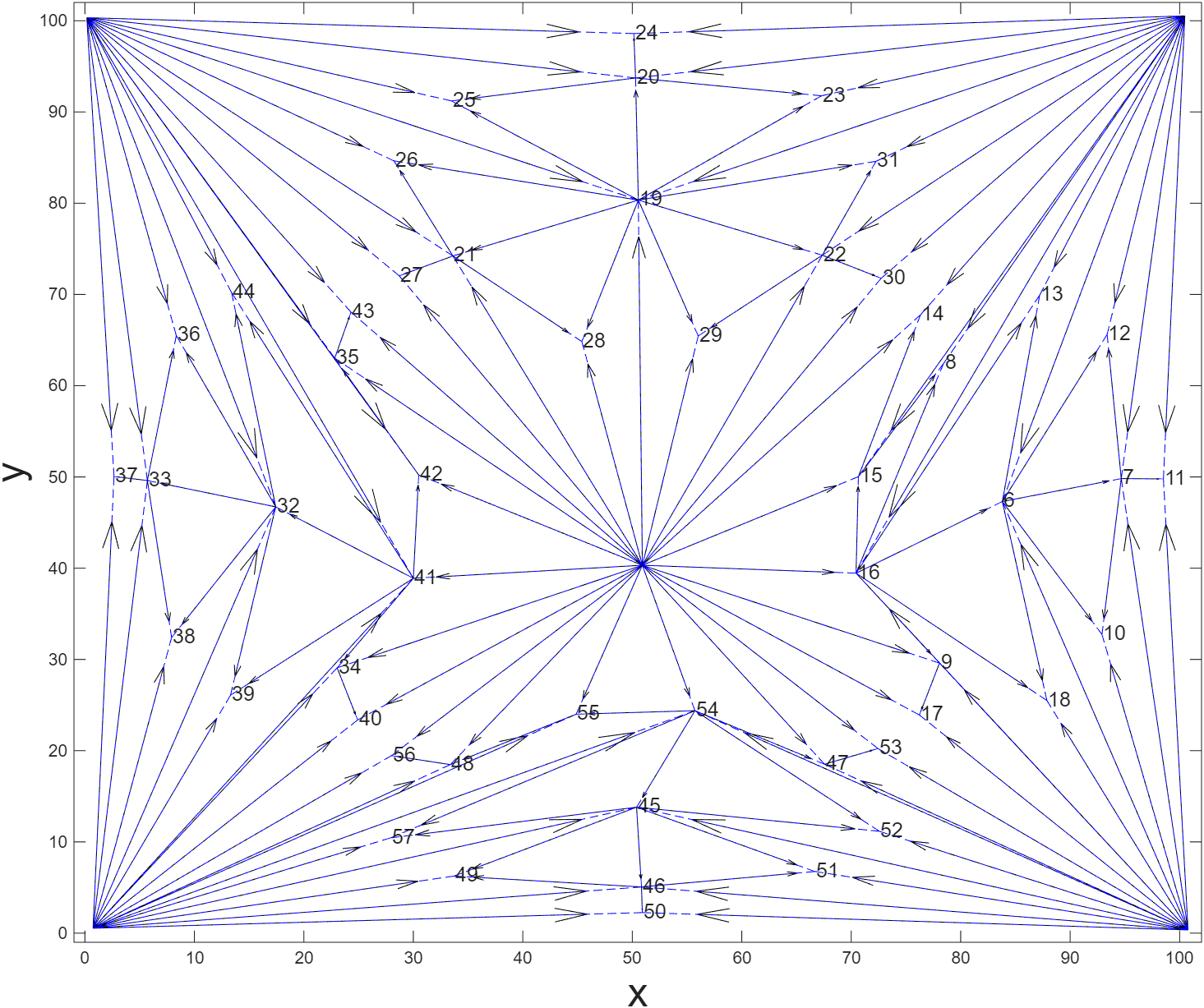}
    \vspace{-0.25cm}
    \caption{The initial formation of the agent team and the communication links.}
    \label{InitialFormation}
\end{figure}
\begin{figure*}
    \centering
    \includegraphics[width=0.98\textwidth]{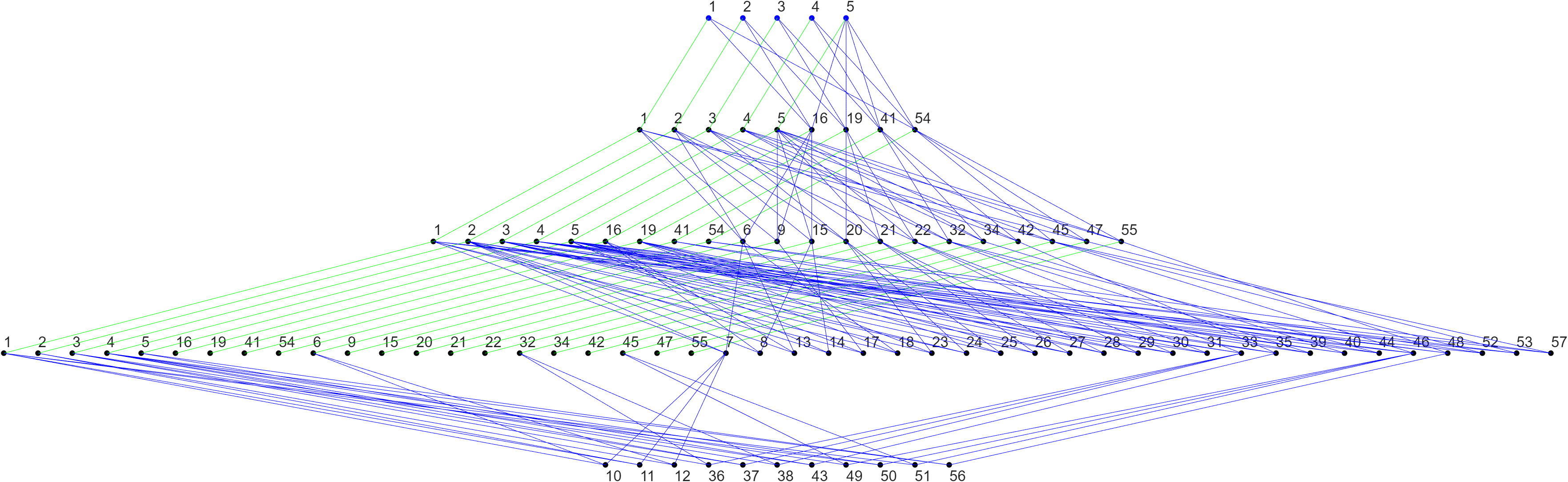}
    \vspace{-0.25cm}
    \caption{The DNN structure consistent with the agents' initial formation in Fig. \ref{InitialFormation}.}
    \label{DNNStructure}
\end{figure*}

\begin{figure}
    \centering
    \includegraphics[width=0.48\textwidth]{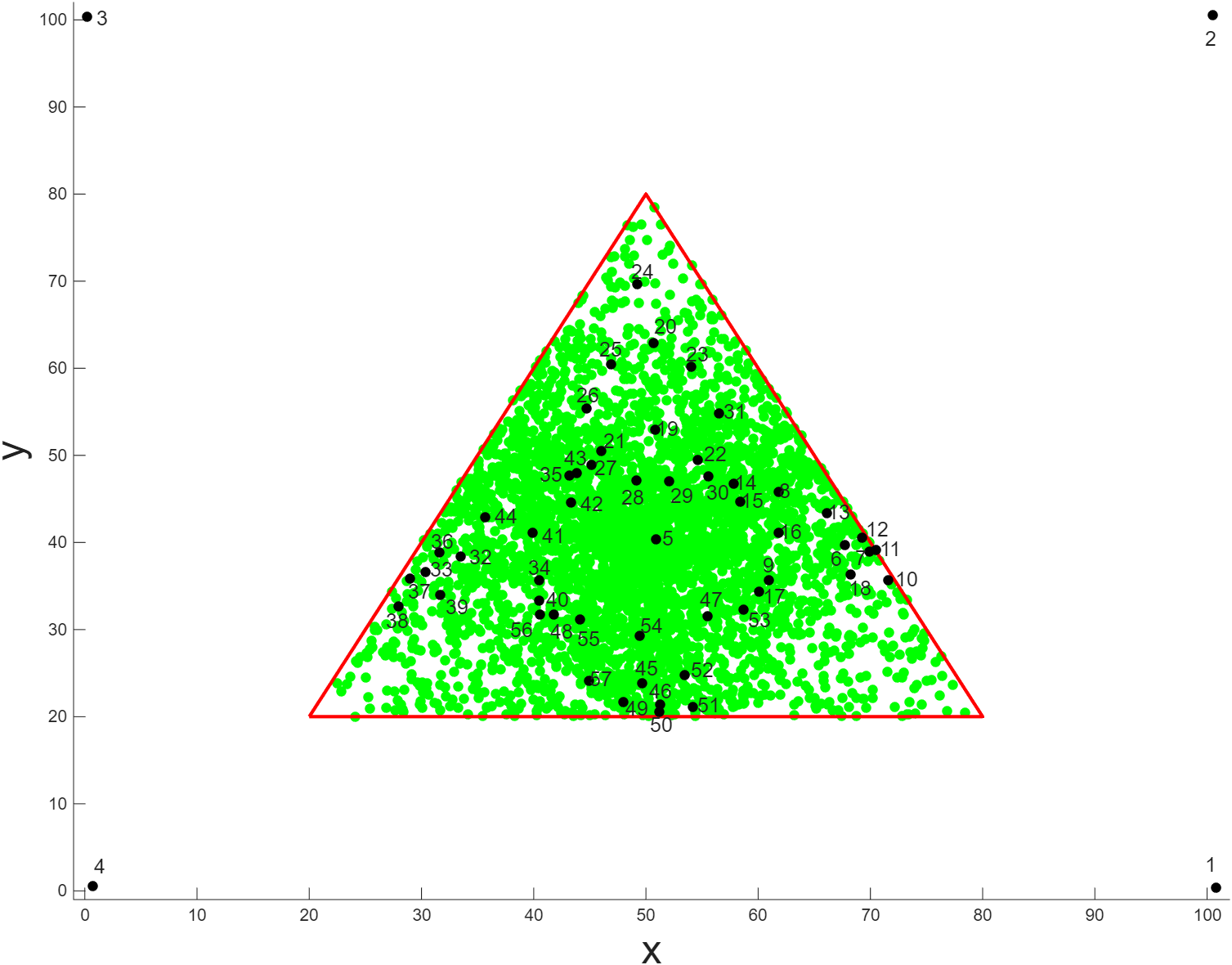}
    \vspace{-0.25cm}
    \caption{The initial formation of the agent team and the communication links.}
    \label{DesiredFormation}
\end{figure}

To establish asymptotic convergence, we introduce $\mathbf{p}_i$ as the desired position of each agent $i\in\mathcal{V}$, which is constant when the target set $\mathcal{D}$ is stationary. The desired positions $\mathbf{p}_i$ are known to all anchored agents $i\in\mathcal{V}_0$. In contrast, the desired positions of follower agents $i\in\mathcal{V}\setminus\mathcal{V}_0$ are not locally available to them. Nevertheless, the quantities $\mathbf{p}_i$ are used solely as analytical constructs to characterize the decentralized convergence of the follower dynamics.

\begin{definition}[Desired communication triangle]
For any agent $b_i\in\mathcal{V}\setminus\mathcal{V}_0$, let
$\mathcal{N}(b_i)=\{b_{i_1}, b_{i_2}, b_{i_3}\}$ denote its in-neighbor set.
The desired communication triangle of $b_i$ is defined as
\begin{equation}\label{tbi}
\tilde{\mathcal{T}}_{b_i}\triangleq
\mathrm{conv}\big\{\mathbf{p}_{b_{i_1}},\,\mathbf{p}_{b_{i_2}},\,\mathbf{p}_{b_{i_3}}\big\},
\end{equation}
i.e., the convex hull of the in-neighbor positions.
\end{definition}

\begin{definition}[Induced target subset]
Given the environmental target set $\mathcal{D}$, the subset of targets covered by
$\tilde{\mathcal{T}}_{b_i}$ is defined as
\begin{equation}\label{dbi}
\tilde{\mathcal{D}}_{b_i}\triangleq
\big\{j\in\mathcal{D}:\ \mathbf{d}_j\in\tilde{\mathcal{T}}_{b_i}\big\}.
\end{equation}
\end{definition}
The desired position of every agent $b_i\in \mathcal{V}\setminus \mathcal{V}_0$ is obtained by
\begin{equation}\label{eq:desired_p}
\tilde{\mathbf{h}}_{b_i}=
\begin{cases}
\dfrac{1}{\left|\tilde{\mathcal{D}}_{b_i}\right|}
\displaystyle\sum_{j\in\tilde{\mathcal{D}}_{b_i}} \mathbf{d}_j[t],
& \tilde{\mathcal{D}}_{b_i}\neq\emptyset,\\[8pt]
\dfrac{1}{3}\displaystyle\sum_{b_j\in\mathcal{N}(b_i)} \mathbf{p}_{b_j},
& \tilde{\mathcal{D}}_{b_i}=\emptyset.
\end{cases}
\end{equation}

Algorithm~\ref{desiredpositions} provides an abstract representation of the environmental target set $\mathcal{D}$ by assigning $N$ desired positions $\mathbf{p}_{b_1},\ldots,\mathbf{p}_{b_N}$ to the agent set.
\begin{definition}
Let $\mathcal{V}=\{b_1,\ldots,b_N\}$. Define
\begin{equation}
\mathbf{z}
=\mathrm{vec}\!\left(
\begin{bmatrix}
\mathbf{p}_{b_1} & \cdots & \mathbf{p}_{b_N}
\end{bmatrix}^{\!T}
\right)
\in \mathbb{R}^{2N}.
\end{equation}
\end{definition}

\begin{definition}
Let $\mathcal{V}_l$ be defined as in \eqref{vl}. Define
\begin{equation}
\mathbf{z}_l
=\mathrm{vec}\!\left(
\begin{bmatrix}
\mathbf{p}_{b_{P_{l-1}+1}} & \cdots & \mathbf{p}_{b_{P_l}}
\end{bmatrix}^{\!T}
\right)
\in \mathbb{R}^{2N_l}.
\end{equation}
\end{definition}

\begin{definition}
For each agent $i\in\mathcal{V}\setminus \mathcal{V}_0$, let $\mathcal{N}(i)$ denote its set of communication in-neighbors, and let $\tilde{\mathcal{T}}_{i}$ denote the associated desired communication triangle.
The center of the goal state $\tilde{g}_{i}$, corresponding to a cell enclosing $\tilde{\mathbf{h}}_i$, defines ${\mathbf{p}}_i$ and is expressed as the convex combination
\begin{equation}\label{desiredweight}
{\mathbf{p}}_i
=\sum_{j\in \mathcal{N}(i)}\tilde{w}_{i,j}\mathbf{p}_{j},\qquad \forall i\in \mathcal{V}\setminus \mathcal{V}_0,
\end{equation}
where $\tilde{w}_{i,j}\in \mathcal{W}_i$.
\end{definition}

\begin{definition}
For each agent $b_i\in\mathcal{V}$, let $\mathcal{N}(b_i)=\{b_{i_1}, b_{i_2}, b_{i_3}\}$ denote its set of communication in-neighbors, and let
\[
{\mathcal{T}}_{b_i}[t]=\mathrm{conv}\big\{\mathbf{r}_{b_{i_1}}[t],\,\mathbf{r}_{b_{i_2}}[t],\,\mathbf{r}_{b_{i_3}}[t]\big\}\]
denote the communication triangle.
The center of the goal state ${g}_{b_i}\in\mathcal{S}_{b_i}$, corresponding to a cell enclosing ${\mathbf{h}}_{b_i}$, is denoted by $\hat{\mathbf{c}}_i$ and expressed as the convex combination
\begin{equation}\label{desiredweight}
{\mathbf{r}}_{b_i}
=\sum_{j\in \mathcal{N}\left(b_i\right)}\hat{w}_{b_i,j}{\mathbf{r}}_{j}
\end{equation}
where $\hat{w}_{b_i,j}\in \mathcal{W}_i$
\end{definition}

\begin{definition}
    We define $\tilde{\mathbf{\Gamma}}=\left[\tilde{\Gamma}_{i,j}\right]$ with the $(i,j)$ entry
    \begin{equation}
        \tilde{\Gamma}_{i,j}[t]=\begin{cases}
1, & b_i\in\mathcal{V}_0,\ i=j,\\
\tilde{w}_{b_i,b_j}[t], & b_i\in\mathcal{V}\setminus\mathcal{V}_0,\ b_j\in\mathcal{N}(b_i),\\
0, & \text{otherwise}.
\end{cases}
    \end{equation}
\end{definition}
\begin{definition}
    We define $\hat{\mathbf{\Gamma}}=\left[\hat{\Gamma}_{i,j}\right]$ with the $(i,j)$ entry
    \begin{equation}
        \hat{\Gamma}_{i,j}[t]=\begin{cases}
1, & b_i\in\mathcal{V}_0,\ i=j,\\
\hat{w}_{b_i,b_j}[t], & b_i\in\mathcal{V}\setminus\mathcal{V}_0,\ b_j\in\mathcal{N}(b_i),\\
0, & \text{otherwise}.
\end{cases}
    \end{equation}
\end{definition}
\begin{algorithm}
  \caption{Environmental Target Representation by $N$ points}\label{desiredpositions}
  \begin{algorithmic}[1]
        \State \textit{Get:} Target set $\mathcal{D}$ and reference position of $\mathcal{V}_0$'s agents, denoted by $\mathbf{a}_{b_1}$ through $\mathbf{a}_{N_0}$, the DNN structure. 
        \State \textit{Obtain:} Agents' desired positions $\mathbf{p}_{b_1}$ through $\mathbf{p}_{b_N}$.
        \For{\texttt{<$ l=0,\cdots,M$>}}
            \If{$l=0$}
                \For{\texttt{<$i=1,\cdots,N_0$>}}
                    \State $\mathbf{p}_{b_i}=\mathbf{a}_{b_i}$.
                \EndFor 
            \Else
                \For{\texttt{<$i=P_{l-1},\cdots,P_l$>}}
                    \State Assign $\tilde{\mathcal{T}}_{b_i}$,   by \eqref{tbi}, and $\tilde{\mathcal{D}}_{b_i}$, by \eqref{dbi}.
                    \State Assign $\mathbf{p}_{b_i}$ by \eqref{eq:desired_p}.
                \EndFor 
            \EndIf
        \EndFor         
  \end{algorithmic}
\end{algorithm}
Matrices $\mathbf{\Gamma}$, $\hat{\mathbf{\Gamma}}$, and $\tilde{\mathbf{\Gamma}}$
share the same strictly lower block–triangular structure with an identity block
in the $(0,0)$ position and a zero last block column. Specifically, for
$l,h\in\mathcal{M}$,
\begin{equation}
\Gamma_{l,h} =
\begin{cases}
\mathbf{I}, & l=h=0,\\
\Gamma_{l,h}, & 0\le h < l \le M,\\
\mathbf{0}, & \text{otherwise},
\end{cases}
\end{equation}
where $\Gamma\in\{\mathbf{\Gamma},\hat{\mathbf{\Gamma}},\tilde{\mathbf{\Gamma}}\}$.











\begin{proposition}
Given $\mathbf{z}_0$, the desired configuration of the agent team satisfies
\begin{equation}\label{zl}
\mathbf{z}_l
=\sum_{h=0}^{l-1}\tilde{\mathbf{\Gamma}}_{l,h}\mathbf{z}_h,
\qquad
\forall\, l\in \mathcal{M}\setminus\{0\}.
\end{equation}
\end{proposition}

\begin{proof}
Following Algorithm~\ref{desiredpositions}, each agent position satisfies
\begin{equation}
\mathbf{p}_{b_i}
=\sum_{j\in\mathcal{N}(b_i)} \tilde{w}_{b_i,j}\mathbf{p}_j .
\end{equation}
Stacking the agent positions yields
\begin{equation}
\mathbf{z}=\tilde{\mathbf{\Gamma}}\mathbf{z},
\end{equation}
from which the recursive relation \eqref{zl} follows directly.
\end{proof}
    

\begin{theorem}
    Let each AOC agent $b_i \in \mathcal{V}\setminus \mathcal{V}_0$ be able to move from any triangle associated with a state $s\in \mathcal{S}_{b_i}$ to the centroid of the triangle associated with its optimal next state $\pi_{b_i}^*(s)\in \mathcal{S}_{b_i}$. Then, for every $b_i\in \mathcal{V}\setminus \mathcal{V}_0$, the desired actual position $\mathbf{r}_{b_i}[t]$ converges asymptotically to $\mathbf{p}_{b_i}$.

\end{theorem}
\begin{proof}
Under the assumptions of the theorem, the subgroup dynamics satisfy
\begin{equation}\label{yld}
\mathbf{y}_l[t+1]
=
\sum_{h=0}^{l-1}\mathbf{\Gamma}_{l,h}\mathbf{y}_h[t],
\qquad
\forall\, l\in \mathcal{M}\setminus\{0\}.
\end{equation}

For $l=1$, $\mathbf{y}_0=\mathbf{z}_0$ is constant. Hence,
$\tilde{\mathcal{D}}_{b_i}=\mathcal{D}_{b_i}$ and
$\tilde{g}_{b_i}=g_{b_i}\in\mathcal{S}(b_i)$ define fixed goal states for all
$b_i\in\mathcal{V}_1$.
By the MDP framework in Section~\ref{Training the DNN Weights}, each agent
$b_i\in\mathcal{V}_1$ converges to $\mathbf{p}_{b_i}$, the center of $g_{b_i}$
enclosing $\tilde{\mathbf{h}}_{b_i}=\mathbf{h}_{b_i}$, implying
$\mathbf{y}_1\to\mathbf{z}_1$. 
Assume for some $l\ge2$ that $\mathbf{y}_{l-1}=\mathbf{z}_{l-1}$.
Then, for all $b_i\in\mathcal{V}_l$, the data sets
$\mathcal{D}_{b_i}[t]\to\tilde{\mathcal{D}}_{b_i}$ and
$\mathbf{h}_{b_i}[t]\to\tilde{\mathbf{h}}_{b_i}$, which implies convergence of
the associated goal states $g_{b_i}\to\tilde{g}_{b_i}$.
Consequently, $\mathbf{\Gamma}_{l,h}[t]\to\tilde{\mathbf{\Gamma}}_{l,h}$ for
$h=0,\ldots,l-1$, and \eqref{yld} yields $\mathbf{y}_l[t]\to\mathbf{z}_l$. 
By induction, $\mathbf{y}_l[t]\to\mathbf{z}_l$ for all
$l\in\mathcal{M}\setminus\{0\}$, and therefore
$\mathbf{r}_i[t]\to\mathbf{p}_i$ for all
$i\in\mathcal{V}\setminus\mathcal{V}_0$.
\end{proof}

\begin{figure}
    \centering
    \includegraphics[width=0.48\textwidth]{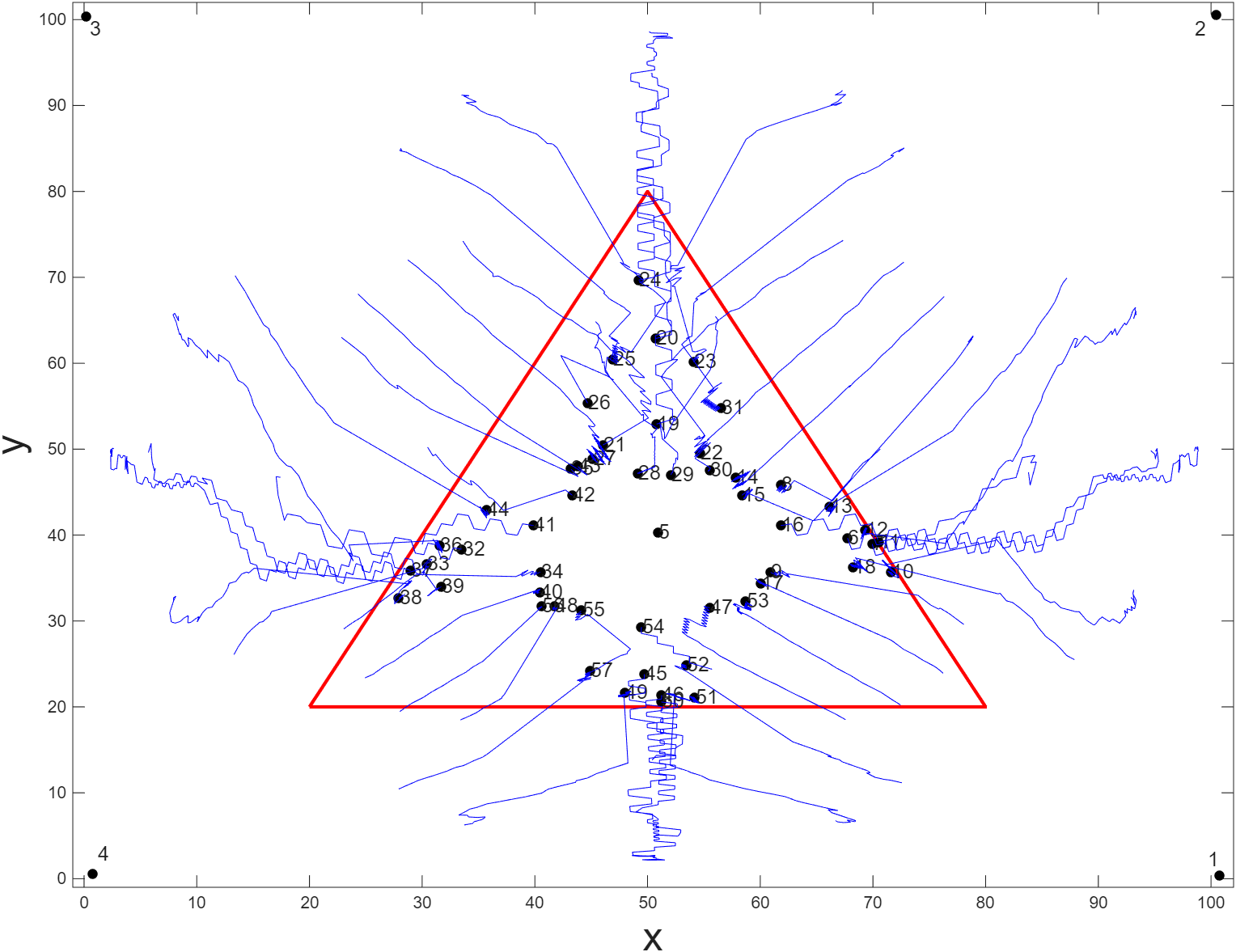}
    \vspace{-0.25cm}
    \caption{Agent paths under the single-step reachability assumption.
All agents $i\in\mathcal{V}$ asymptotically converge to their desired positions $\mathbf{p}_i$.}
    \label{AgentsPathsPrecise}
\end{figure}
\section{Simulation Results}\label{Simulation Results}
We consider an $57$-agent system with the initial configuration shown in Fig.~\ref{InitialFormation}. 
Based on the reference formation, $\mathcal{V}_B=\{1,2,3,4\}$ and $b_5=5$ denote the boundary and core agents, respectively, while all remaining agents are classified as interior. 
The resulting inter-agent communication structure is encoded by the DNN shown in Fig.~\ref{DNNStructure} (arrows in  Fig.~\ref{InitialFormation}), constructed using the framework of Section~\ref{Structuring of the DNN and Specifying the DNN Weights}. It is desired that the multi-agent system cover the triangular domain shown in Fig.~\ref{DesiredFormation}, where the environmental target data defined by $\mathcal{D}$ are shown in green. The desired positions of the agent team, denoted $\mathbf{p}_1$ through $\mathbf{p}_{57}$, are shown by black. To define the state space, we choose $M_i=35$, for every $i\in \mathcal{V}\setminus \mathcal{V}_0$, which in turn implies that $\left|\mathcal{S}_i^C\right|=35^2=1225$.
\subsection{Evolution under Finite-Time Reachability of $\mathbf{h}_i$}

In this section, we assume that each agent $i\in \mathcal{V}\setminus\mathcal{V}_0$ can reach $\mathbf{h}_i[t]$ in a single time step, which implies $\mathcal{P}_i=\mathbf{h}_i[t]$.
Under this assumption, the resulting agent paths are shown in Fig.~\ref{AgentsPathsPrecise}, where all agents $i\in\mathcal{V}$ asymptotically converge to their desired positions $\mathbf{p}_i$.

To further illustrate convergence, Figs.~\ref{PreciseTraj}(a)–(b) show the $x$- and $y$-components of the actual and desired positions of agent~29, $\mathbf{r}_{29}[t]$ and $\mathbf{p}_{29}$, respectively, as functions of discrete time $t$. The trajectories demonstrate rapid convergence, with $\mathbf{r}_{29}[t]$ reaching $\mathbf{p}_{29}$ in fewer than $40$ time steps.
\begin{figure}[h]
\centering
\subfigure[]{\includegraphics[width=0.49\linewidth]{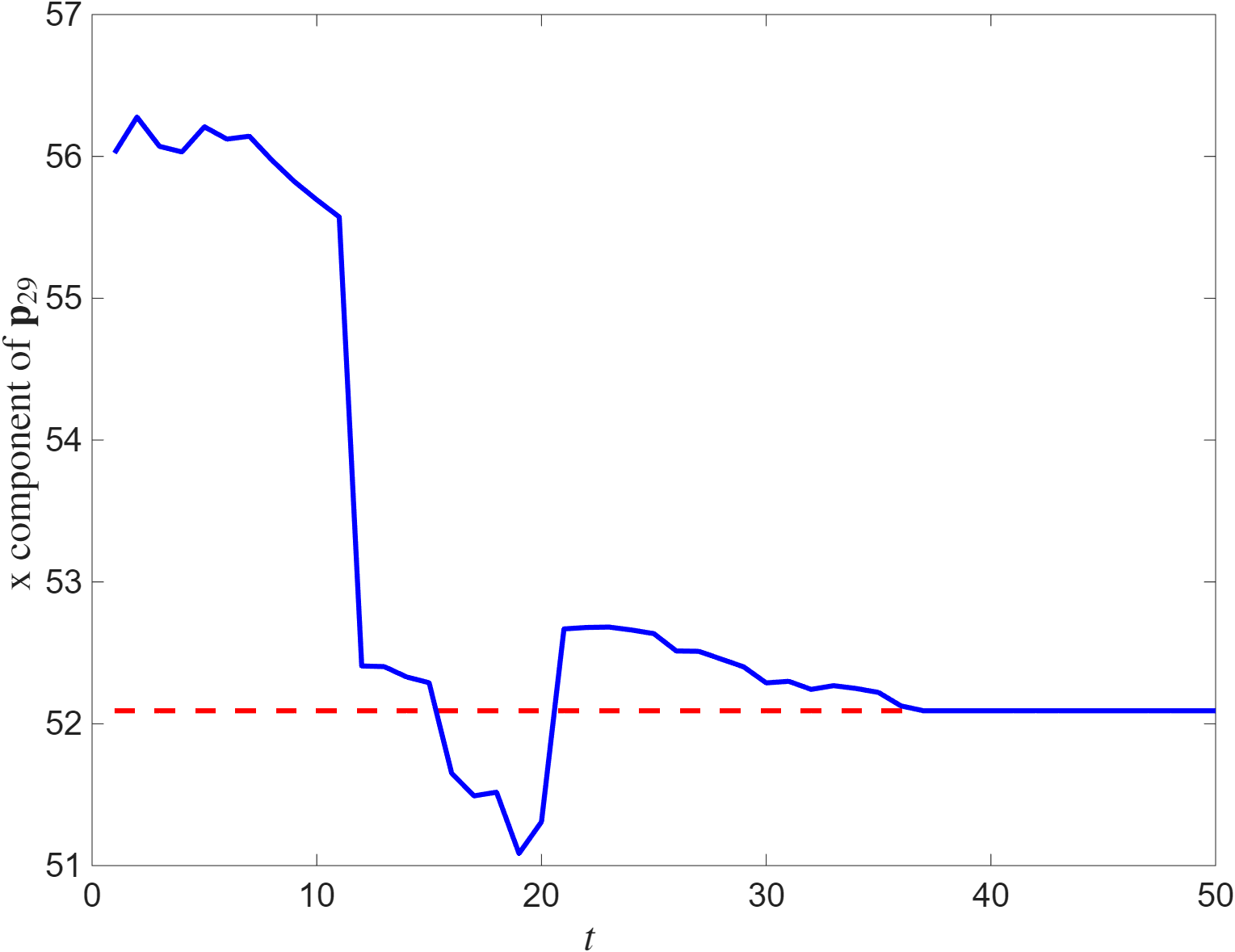}}
\subfigure[]{\includegraphics[width=0.49\linewidth]{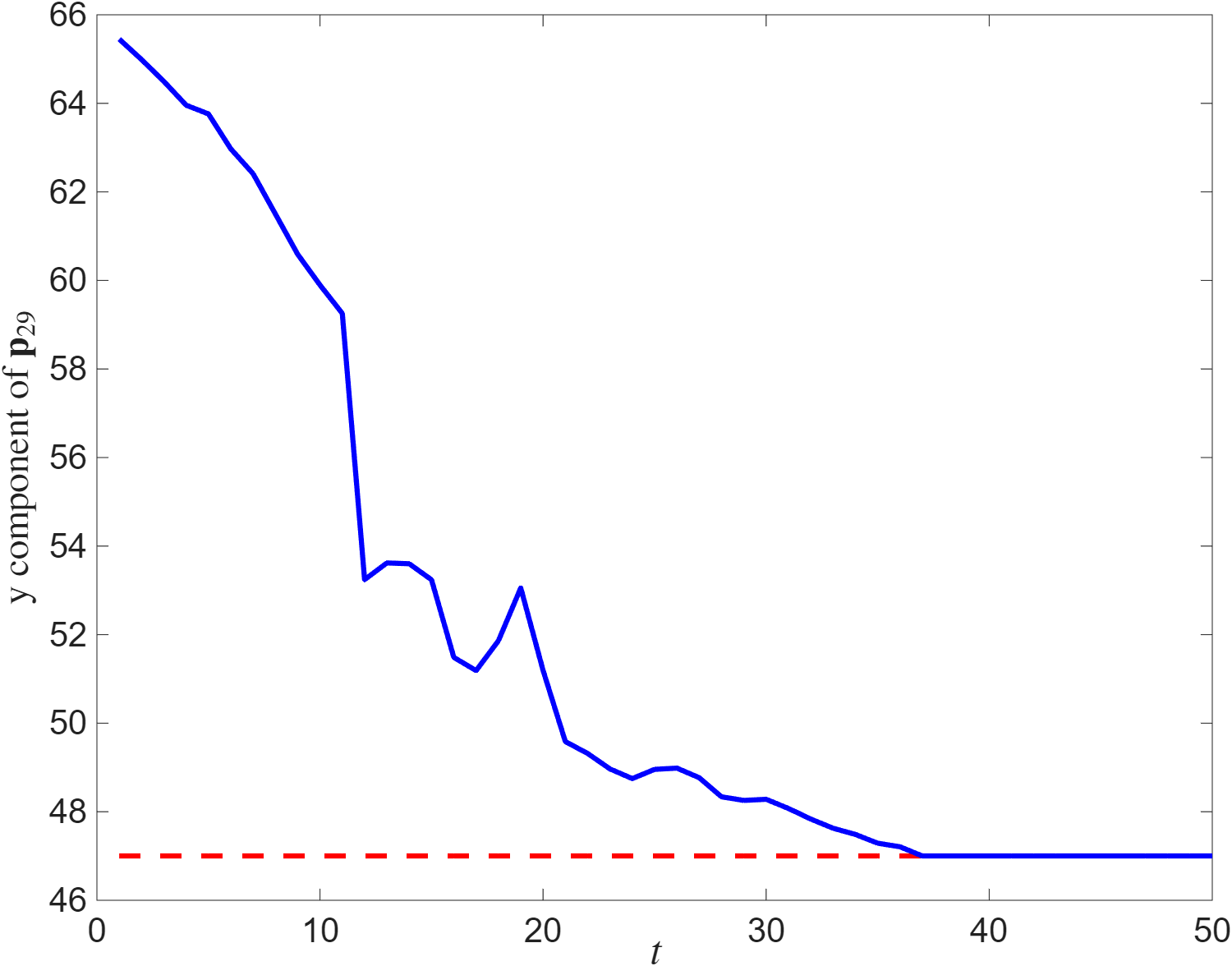}}
\vspace{-0.4cm}
\caption{Time evolution of the $x$- and $y$-components of the actual and desired positions of agent~29, illustrating convergence to $\mathbf{p}_{29}$.
}
\label{PreciseTraj}
\end{figure}

\subsection{Evolution under AOC Assumption}
In this section, we analyze the agents’ evolution under Assumption~\ref{currentinside}, wherein the desired position $\mathbf{p}_i$ is constrained to lie within an edge-aligned triangular region strictly contained in $\mathcal{T}_i(t)$, guaranteeing $\eta = 0.05$ for all $i \in \mathcal{V}\setminus\mathcal{V}_0$. Under this condition, the resulting closed-loop trajectories of all agents in the $x$–$y$ plane are shown in Fig.~\ref{AOC1}, demonstrating coordinated motion and spatial containment. Moreover, Fig.~\ref{AOC2} depicts the temporal evolution of the $x$- and $y$-components of the desired trajectory for agent~43, illustrating precise tracking behavior over discrete time $t$.

\begin{figure}
    \centering
    \includegraphics[width=0.48\textwidth]{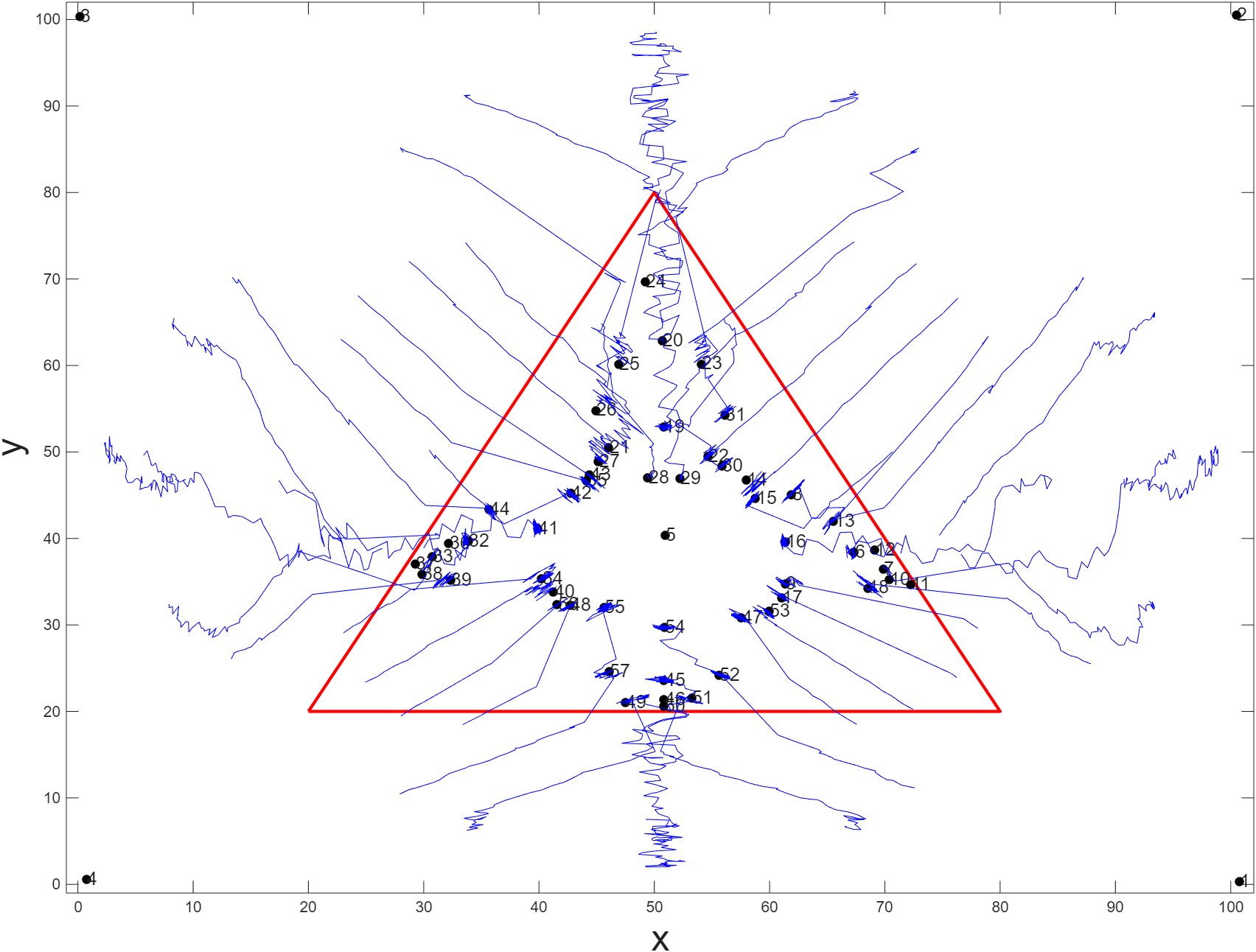}
    \vspace{-0.25cm}
    \caption{Agent trajectories under the single-step reachability assumption and Assumption~\ref{currentinside}. 
All agents $i\in\mathcal{V}$ asymptotically converge to their desired positions $\mathbf{p}_i$.
}
    \label{AOC1}
\end{figure}

\begin{figure}[h]
\centering
\subfigure[]{\includegraphics[width=0.49\linewidth]{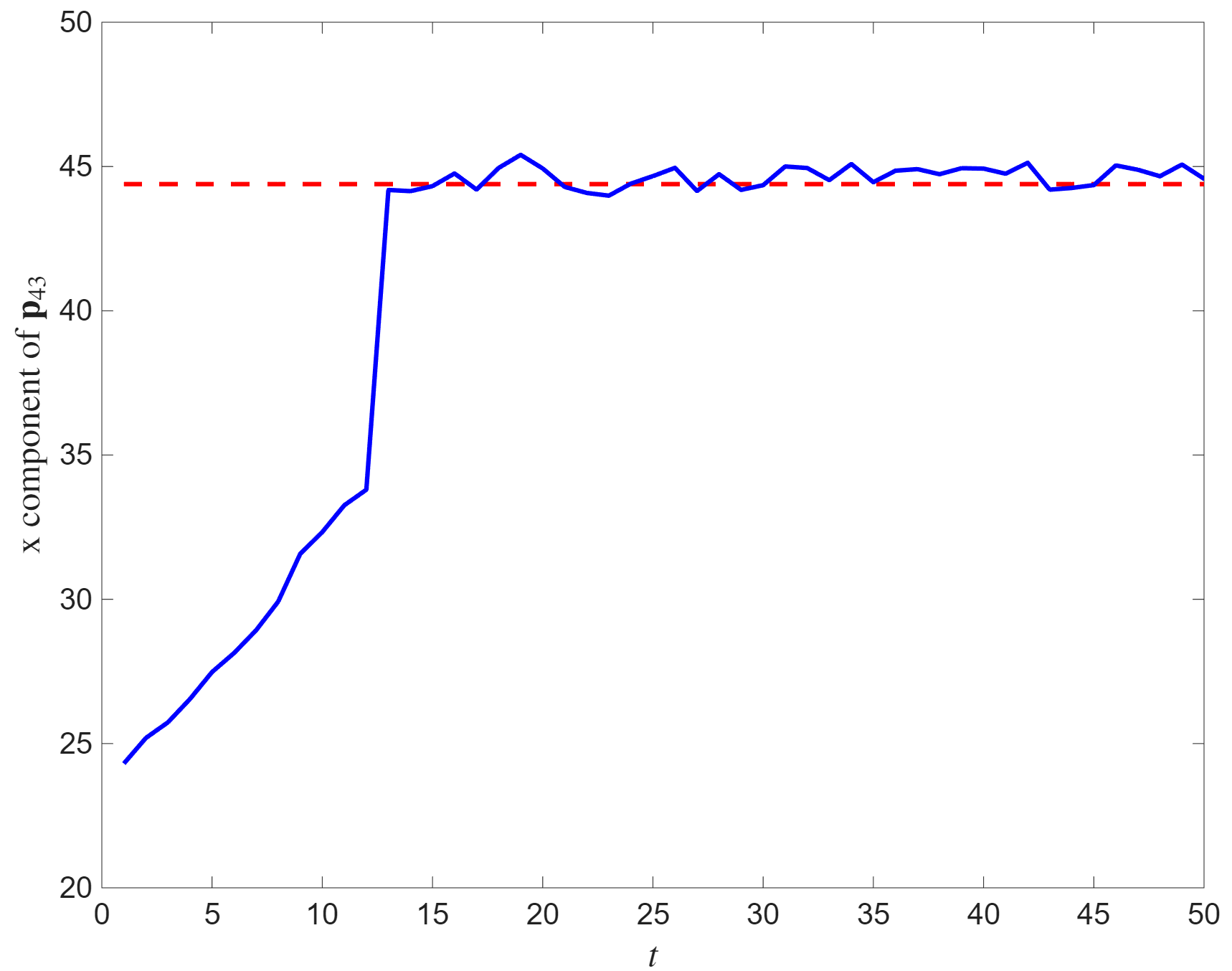}}
\subfigure[]{\includegraphics[width=0.49\linewidth]{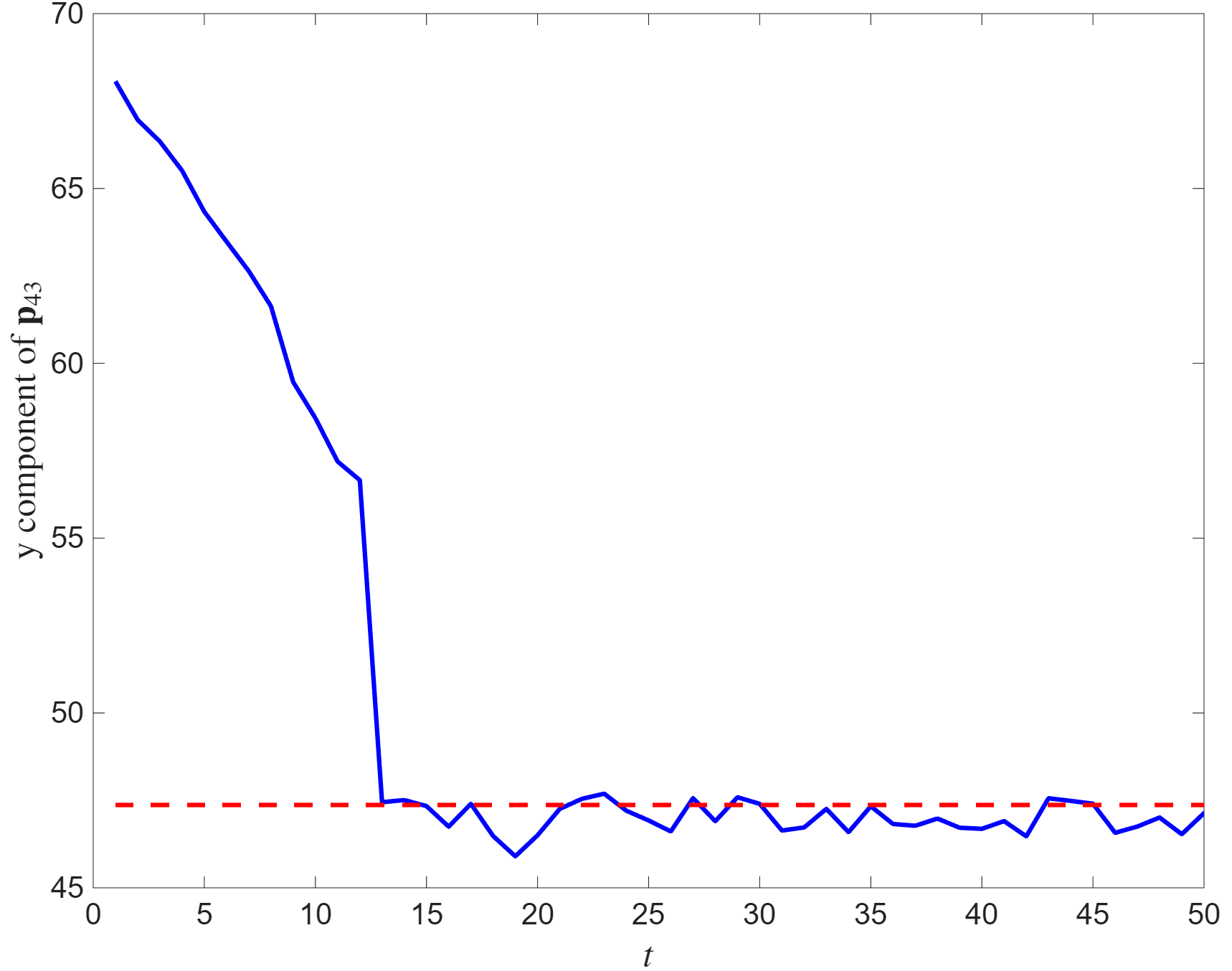}}
\vspace{-0.4cm}
\caption{Time evolution of the $x$- and $y$-components of the actual and desired positions of agent~43, illustrating convergence to $\mathbf{p}_{43}$.
}
\label{AOC2}
\end{figure}

\section{Conclusion}
\label{Conclusion}
This paper presented a structured learning–based framework for decentralized coordination and ground coverage in multi-agent systems, in which inter-agent communication is encoded through a geometrically induced deep neural network. By exploiting the reference formation, agents are systematically classified into boundary and interior groups, yielding a hierarchical communication architecture with explicitly constrained and interpretable communication weights. These weights are selected from finite sets and governed by a decentralized Markov decision process, ensuring normalized interactions and well-posed local decision making. Within this framework, convergence of agent trajectories to desired goal configurations associated with environmental target data was established under AOC assumptions. Numerical simulations validate the proposed policy-based decentralized coverage strategy and demonstrate its ability to capture geometric structure and achieve effective coverage of complex domains.

     \bibliographystyle{plain}        
\bibliography{autosam, reference, citation, CoverageReference}           
\begin{IEEEbiography}[{\includegraphics[width=1in,height=1.25in,clip,keepaspectratio]{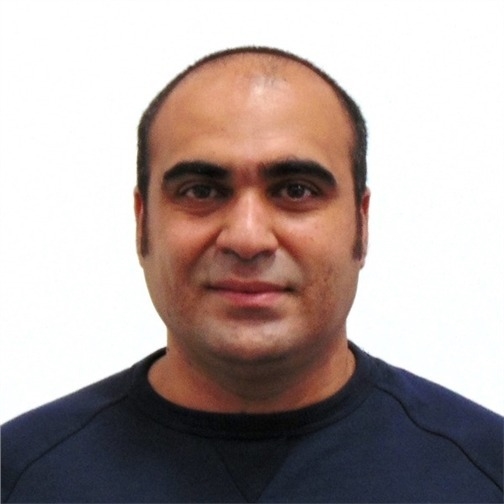}}]
{\textbf{Hossein Rastgoftar}} an Assistant Professor at the University of Arizona. Prior to this, he was an adjunct Assistant Professor at the University of Michigan from 2020 to 2021. He was also an Assistant Research Scientist (2017 to 2020) and a Postdoctoral Researcher (2015 to 2017) in the Aerospace Engineering Department at the University of Michigan Ann Arbor. He received the B.Sc. degree in mechanical engineering-thermo-fluids from Shiraz University, Shiraz, Iran, the M.S. degrees in mechanical systems and solid mechanics from Shiraz University and the University of Central Florida, Orlando, FL, USA, and the Ph.D. degree in mechanical engineering from Drexel University, Philadelphia, in 2015. 
\end{IEEEbiography}

\end{document}